\documentclass[11pt]{llncs}
\usepackage{graphicx}

\usepackage{amsfonts}
\usepackage{amsmath}
\usepackage{amssymb}
\usepackage{fullpage}

\usepackage{booktabs} 
\usepackage[ruled]{algorithm2e} 

\SetAlFnt{\small}
\SetAlCapFnt{\small}
\SetAlCapNameFnt{\small}
\SetAlCapHSkip{0pt}
\IncMargin{-\parindent}

\usepackage[margin=1in]{geometry}

\newtheorem{observation}{Observation}
\newtheorem{myclaim}{Claim}

\def\eps{\varepsilon}

\def\Diam{D}

\def\eps{\varepsilon}
\newcommand{\opt}{{\sc Optimal}}
\newcommand{\rightm}{{\sc AlmostRightmost}}
\newcommand{\rand}{{\sc Random}}

\def\Exp{\mathbb{E}}
\def\Prob{\mathbb{P}\mathrm{r}}
\def\reals{\mathbb{R}}
\def\SC{cost}

\begin{document}

\title{Strategyproof Facility Location in Perturbation Stable Instances%
\thanks{This work was supported by the Hellenic Foundation for Research and Innovation (H.F.R.I.) under the ``First Call for H.F.R.I. Research Projects to support Faculty members and Researchers' and the procurement of high-cost research equipment grant'', project BALSAM, HFRI-FM17-1424
}}

\author{Dimitris Fotakis  \and
Panagiotis Patsilinakos}

\authorrunning{Dimitris Fotakis and Panagiotis Patsilinakos}
%
\institute{School of Electrical and Computer Engineering\\ National Technical University of Athens, 15780 Athens, Greece\\[4pt]
\email{fotakis@cs.ntua.gr}, \email{patsilinak@corelab.ntua.gr}}

\maketitle

\begin{abstract}
We consider $k$-Facility Location games, where $n$ strategic agents report their locations on the real line, and a mechanism maps them to $k \geq 2$ facilities. Each agent seeks to minimize her distance to the nearest facility. We are interested in (deterministic or randomized) strategyproof mechanisms without payments that achieve a reasonable approximation ratio to the optimal social cost of the agents. To circumvent the inapproximability of $k$-Facility Location by deterministic strategyproof mechanisms, we restrict our attention to perturbation stable instances. An instance of $k$-Facility Location on the line is \emph{$\gamma$-perturbation stable} (or simply, \emph{$\gamma$-stable}), for some $\gamma \geq 1$, if the optimal agent clustering is not affected by moving any subset of consecutive agent locations closer to each other by a factor at most $\gamma$. We show that the optimal solution is strategyproof in $(2+\sqrt{3})$-stable instances whose optimal solution does not include any singleton clusters, and that allocating the facility to the agent next to the rightmost one in each optimal cluster (or to the unique agent, for singleton clusters) is strategyproof and $(n-2)/2$-approximate for $5$-stable instances (even if their optimal solution includes singleton clusters). On the negative side, we show that for any $k \geq 3$ and any $\delta > 0$, there is no deterministic anonymous mechanism that achieves a bounded approximation ratio and is strategyproof in $(\sqrt{2}-\delta)$-stable instances. We also prove that allocating the facility to a random agent of each optimal cluster is strategyproof and $2$-approximate in $5$-stable instances. To the best of our knowledge, this is the first time that the existence of deterministic (resp. randomized) strategyproof mechanisms with a bounded (resp. constant) approximation ratio is shown for a large and natural class of $k$-Facility Location instances.
\end{abstract}


\pagestyle{plain}

\section{Introduction}
\label{s:intro}

We consider \emph{$k$-Facility Location games}, where $k \geq 2$ facilities are placed on the real line based on the preferences of $n$ strategic agents. Such problems are motivated by natural scenarios in Social Choice, where a local  authority plans to build a fixed number of public facilities in an area (see e.g., \cite{Miy01}). The choice of the locations is based on the preferences of local people, or \emph{agents}. Each agent reports her ideal location, and the local authority applies a (deterministic or randomized) \emph{mechanism} that maps the agents' preferences to $k$ facility locations.

Each agent evaluates the mechanism's outcome according to her \emph{connection cost}, i.e., the distance of her ideal location to the nearest facility. The agents seek to minimize their connection cost and may misreport their ideal locations in an attempt of manipulating the mechanism. Therefore, the mechanism should be \emph{strategyproof}, i.e., it should ensure that no agent can benefit from misreporting her location, or even \emph{group strategyproof}, i.e., resistant to coalitional manipulations. 
The local authority's objective is to minimize the \emph{social cost}, namely the sum of agent connections costs. In addition to allocating the facilities in a incentive compatible way, which is formalized by (group) strategyproofness, the mechanism should result in a socially desirable outcome, which is quantified by the mechanism's approximation ratio to the optimal social cost. 

Since Procaccia and Tennenholtz \cite{PT09} initiated the research agenda of \emph{approximate mechanism design without money}, $k$-Facility Location has served as the benchmark problem in the area and its approximability by deterministic or randomized strategyproof mechanisms has been studied extensively in virtually all possible variants and generalizations. For instance, previous work has considered multiple facilities on the line (see e.g., \cite{FT12,FT13,GT17,LWZ09,NST12}) and in general metric spaces \cite{FT10,LSWZ10}), different objectives (e.g., social cost, maximum cost, the $L_2$ norm of agent connection costs \cite{FW11,PT09,FT13}), restricted metric spaces more general than the line (cycle, plane, trees, see e.g., \cite{AFPT09,DFMN12,FM21,GH20,Meir19}), facilities that serve different purposes (see e.g., \cite{KVZ19,LiLYZ20,SV16}), and different notions of private information about the agent preferences that should be declared to the mechanism (see e.g., \cite{ChenFLWYZ20,FLSWZ20,MeiLYZ19} and the references therein). 

Due to the significant research interest in the topic, the fundamental and most basic question of approximating the optimal social cost by strategyproof mechanisms for $k$-Facility Location on the line has been relatively well-understood. For a single facility ($k=1$), placing the facility at the median location is group strategyproof and optimizes the social cost. For two facilities ($k=2$), the best possible approximation ratio is $n-2$ and is achieved by a natural group strategyproof mechanism that places the facilities at the leftmost and the rightmost location \cite{FT12,PT09}. However, for three or more facilities ($k \geq 3$), there do not exist any deterministic anonymous%
\footnote{A mechanism is \emph{anonymous} if its outcome depends only on the agent locations, not on their identities.}
strategyproof mechanisms for $k$-Facility Location with a bounded (in terms of $n$ and $k$) approximation ratio \cite{FT12}. On the positive side, there is a randomized anonymous group strategyproof mechanism%
\footnote{The result of \cite{FT13} applies to the more general setting where the agent connection cost is a nondecreasing concave function of the distance to the nearest facility.}
with an approximation ratio of $n$ \cite{FT13} (see also Section~\ref{s:previous} for a selective list of additional references).

\smallskip\noindent{\bf\boldmath Perturbation Stability in $k$-Facility Location Games.}
Our work aims to circumvent the strong impossibility result of \cite{FT12} and is motivated by the recent success on the design of polynomial-time exact algorithms for perturbation stable clustering instances (see e.g., \cite{AMM17,BHW16,BBG13,BL16,Rough19,Rough20}). An instance of a clustering problem, like $k$-Facility Location (a.k.a. $k$-median in the optimization and approximation algorithms literature), is \emph{$\gamma$-perturbation stable} (or simply, \emph{$\gamma$-stable}), for some $\gamma \geq 1$, if the optimal clustering is not affected by scaling down any subset of the entries of the distance matrix by a factor at most $\gamma$. Perturbation stability was introduced by Bilu and Linial \cite{BL10} and Awasthi, Blum and Sheffet \cite{ABS12} (and has motivated a significant volume of followup work since then, see e.g., \cite{AMM17,BHW16,BL16,Rough20} and the references therein) in an attempt to obtain a theoretical understanding of the superior practical performance of relatively simple clustering algorithms for well known $\mathrm{NP}$-hard clustering problems (such as $k$-Facility Location in general metric spaces). Intuitively, the optimal clusters of a $\gamma$-stable instance are somehow well separated, and thus, relatively easy to identify (see also the main properties of stable instances in Section~\ref{s:stability}). As a result, natural extensions of simple algorithms, like single-linkage (a.k.a. Kruskal's algorithm), can recover the optimal clustering in polynomial time, provided that $\gamma \geq 2$ \cite{AMM17}, and standard approaches, like dynamic programming (resp. local search), work in almost linear time for $\gamma > 2+\sqrt{3}$ (resp. $\gamma > 5$) \cite{ACM0W20}. 

In this work, we investigate whether restricting our attention to stable instances allows for improved strategyproof mechanisms with bounded (and ideally, constant) approximation guarantees for $k$-Facility Location on the line, with $k \geq 2$. We note that the impossibility results of \cite{FT12} crucially depend on the fact that the clustering (and the subsequent facility placement) produced by any deterministic mechanism with a bounded approximation ratio must be sensitive to location misreports by certain agents (see also Section~\ref{s:lower_bound}). Hence, it is very natural to investigate whether the restriction to $\gamma$-stable instances allows for some nontrivial approximation guarantees by deterministic or randomized strategyproof mechanisms for $k$-Facility Location on the line. 

To study the question above, we adapt to the real line the stricter%
\footnote{The notion of $\gamma$-metric stability is ``stricter'' than standard $\gamma$-stability in the sense that the former excludes some perturbations allowed by the latter. Hence, the class of $\gamma$-metric stable instances includes the class of $\gamma$-stable instances. More generally, the stricter a notion of stability is, the larger the class of instances qualified as stable, and the more general the positive results that one gets. Similarly, for any $\gamma' > \gamma \geq 1$, the class of $\gamma$(-metric) stable instances includes the class of $\gamma'$(-metric) instances. Hence, a smaller value of $\gamma$ makes a positive result stronger and more general.}
notion of \emph{$\gamma$-metric stability} \cite{AMM17}, where the definition also requires that the distances form a metric after the $\gamma$-perturbation. In our notion of \emph{linear $\gamma$-stability}, the instances should retain their linear structure after a $\gamma$-perturbation. Hence, a $\gamma$-perturbation of a linear $k$-Facility Location instance is obtained by moving any subset of pairs of consecutive agent locations closer to each other by a factor at most $\gamma \geq 1$. We say that a $k$-Facility Location instance is $\gamma$-stable, if the original instance and any $\gamma$-perturbation of it admit the same unique optimal clustering%
\footnote{As for the optimal centers, in case of ties, the center of an optimal cluster is determined by a fixed deterministic tie-breaking rule, e.g., the center is always the left median point of the cluster.} (see also Definition~\ref{def:stable}).

Interestingly, for $\gamma$ sufficiently large, $\gamma$-stable instances of $k$-Facility Location have additional structure that one could exploit towards the design of strategyproof mechanisms with good approximation guarantees (see also Section~\ref{s:stability}). E.g., each agent location is $\gamma-1$ times closer to the nearest facility than to any location in a different cluster (Proposition~\ref{pr:separation}). Moreover, for $\gamma \geq 2+\sqrt{3}$, the distance between any two consecutive clusters is larger than their diameter (Lemma~\ref{l:separated-clusters}). 

From a conceptual viewpoint, our work is motivated by a reasoning very similar to that discussed by Bilu, Daniely, Linial and Saks \cite{BDLS13} and summarized in  ``\emph{clustering is hard only when it doesn't matter}'' by Roughgarden~\cite{Rough17_lect6}. In a nutshell, we expect that when $k$ public facilities (such as schools, libraries, hospitals, representatives) are to be allocated to some communities (e.g., cities, villages or neighborhoods, as represented by the locations of agents on the real line) the communities are already well formed, relatively easy to identify and difficult to radically reshape by small distance perturbations or agent location misreports. Moreover, in natural practical applications of $k$-Facility Location games, agents tend to misreport ``locally'' (i.e., they tend to declare a different ideal location in their neighborhood, trying to manipulate the location of the local facility), which usually does not affect the cluster formation. In practice, this happens because the agents do not have enough knowledge about locations in other neighborhoods, and because ``large non-local'' misreports are usually easy to identify by combining publicly available information about the agents (e.g., occupation, address, habits, lifestyle). Hence, we believe that the class of $\gamma$-stable instances, especially for relatively small values of $\gamma$, provides a reasonably accurate abstraction of the instances of $k$-Facility Location games that a mechanism designer is more likely to deal with in practice. Thus, we feel that our work takes a small first step towards justifying that (not only clustering but also) strategyproof facility location is hard only when it doesn't matter.

\smallskip\noindent{\bf Contributions and Techniques.}
Our conceptual contribution is that we initiate the study of efficient (wrt. their approximation ratio for the social cost) strategyproof mechanisms for the large and natural class of $\gamma$-stable instances of $k$-Facility Location on the line. 
Our technical contribution is that we show the existence of deterministic (resp. randomized) strategyproof mechanisms with a bounded (resp. constant) approximation ratio for $5$-stable instances and any number of facilities $k \geq 2$. Moreover, we show that the optimal solution is strategyproof for $(2+\sqrt{3})$-stable instances, if the optimal clustering does not include any singleton clusters (which is likely to be the case in virtually all practical applications). 
To provide evidence that restriction to stable instances does not make the problem trivial, we strengthen the impossibility result of Fotakis and Tzamos \cite{FT12}, so that it applies to $\gamma$-stable instances, with $\gamma < \sqrt{2}$. Specifically, we show that that for any $k \geq 3$ and any $\delta > 0$, there do not exist any deterministic anonymous strategyproof mechanisms for $k$-Facility Location on $(\sqrt{2}-\delta)$-stable instances with bounded (in terms of $n$ and $k$) approximation ratio. 

At the conceptual level, we interpret the stability assumption as a prior on the class of true instances that the mechanism should be able to deal with. Namely, we assume that the mechanism has only to deal with 
$\gamma$-stable true instances, 
a restriction motivated by (and fully consistent with) how the stability assumption is used in the literature on efficient algorithms for stable clustering (see e.g., \cite{AMM17,BHW16,BL16,BL10}, where the algorithms are analyzed for stable instances only). 
More specifically, our mechanisms expect as input a declared instance such that in the optimal clustering, the distance between any two consecutive clusters is at least $\frac{(\gamma-1)^2}{2\gamma}$ times larger than the diameters of the two clusters (a.k.a. \emph{cluster-separation} property, see Lemma~\ref{c:separated-clusters}). This condition is necessary (but not sufficient) for $\gamma$-stability and can be easily checked. If the declared instance does not satisfy the cluster-separation property, our mechanisms do not allocate any facilities. Otherwise, our mechanisms allocate $k$ facilities (even if the instance is not stable). 
We prove that for all $\gamma$-stable true instances (with the exact stability factor $\gamma$ depending on the mechanism), if agents can only deviate so that the declared instance satisfies the cluster-separation property (and does not have singleton clusters, for the optimal mechanism), our mechanisms are strategyproof and achieve the desired approximation guarantee. 
Hence, if we restrict ourselves to $\gamma$-stable true instances and to agent deviations that do not obviously violate $\gamma$-stability, our mechanisms should only deal with $\gamma$-stable declared instances, due to strategyproofness. 
On the other hand, if non-stable true instances may occur, the mechanisms cannot distinguish between a stable true instance and a declared instance, which appears to be stable, but is obtained from a non-stable instance through location misreports.


The restriction that the agents of a $\gamma$-stable instance are only allowed to deviate so that the declared instance satisfies the cluster-separation property (and does not have any singleton clusters, for the optimal mechanism) bears a strong conceptual resemblance to the notion of strategyproof mechanisms with \emph{local verification} (see e.g., \cite{ADPP09,AK08,CESY12,Carroll13,FTZ16,FZ13,GL86}), where the set of each agent's allowable deviations is restricted to a so-called \emph{correspondence set}, which typically depends on the agent's true type, but not on the types of the other agents. Instead of restricting the correspondence set of each individual agent independently, we impose a structural condition on the entire declared instance, which restricts the set of the agents' allowable deviations, but in a global and observable sense. As a result, we can actually implement our notion of verification, by checking some simple properties of the declared instance, instead of just assuming that any deviation outside an agent's correspondence set will be caught and penalized (which is the standard approach in mechanisms with local verification \cite{AK08,Carroll13,CESY12}, but see e.g., \cite{ADPP09,FT10} for noticeable exceptions). 

On the technical side, we start, in Section~\ref{s:stability}, with some useful properties of stables instances of $k$-Facility Location on the line. Among others, we show (i) the \emph{cluster-separation} property (Lemma~\ref{l:separated-clusters}), which states that in any $\gamma$-stable instance, the distance between any two consecutive clusters is at least $\frac{(\gamma-1)^2}{2\gamma}$ times larger than their diameters; and (ii) the so-called \emph{no direct improvement from singleton deviations} property (Lemma~\ref{3_stab_lemma}), i.e., that in any $3$-stable instance, no agent who deviates to a location, which becomes a singleton cluster in the optimal clustering of the resulting instance, can improve her connection cost through the facility of that singleton cluster.  

In Section~\ref{s:optimal}, we show that for $(2+\sqrt{3})$-stable instances whose optimal clustering does not include any singleton clusters, the optimal solution is strategyproof (Theorem~\ref{th:optimal}). For the analysis, we observe that since placing the facility at the median location of any fixed cluster is strategyproof, a misreport cannot be profitable for an agent, unless it results in a different optimal clustering. The key step is to show that for $(2+\sqrt{3})$-stable instances without singleton clusters, a profitable misreport cannot change the optimal clustering, unless the instance obtained from the misreport violates the cluster-separation property. To the best of our knowledge, the idea of penalizing (and thus, essentially forbidding) a whole class of potentially profitable misreports by identifying how they affect a key structural property of the original instance, which becomes possible due to our restriction to stable instances, has not been used before in the design of strategyproof mechanisms for $k$-Facility Location (see also the discussion above about resemblance to mechanisms with verification). 

We should also motivate our restriction to stable instances without singleton clusters in their optimal clustering. So, let us consider the rightmost agent $x_j$ of an optimal cluster $C_i$ in a $\gamma$-stable instance $\vec{x}$. No matter the stability factor $\gamma$, it is possible that $x_j$ performs a so-called \emph{singleton deviation}. Namely, $x_j$ deviates to a remote location $x'$ (potentially very far away from any location in $\vec{x}$), which becomes a singleton cluster in the optimal clustering of the resulting instance. Such a singleton deviation might cause cluster $C_i$ to merge with (possibly part of the next) cluster $C_{i+1}$, which in turn, might bring the median of the new cluster much closer to $x_j$ (see also Fig.~\ref{f:singletonClusters} in Section~\ref{s:stability}). It is not hard to see that if we stick to the optimal solution, where the facilities are located at the median of each optimal cluster, there are $\gamma$-stable instances%
\footnote{E.g., let $k = 2$ and consider the $\Theta(\gamma)$-stable instance $(0, 1-\eps, 1, 6\gamma, 6\gamma+\eps, 6\gamma+1, 6\gamma+1+\eps, 6\gamma+2)$, for any $\gamma \geq 1$. Then, the agent at location $6\gamma$ can decrease its connection cost (from $1$) to $\eps$ by deviating to location $(6\gamma)^2$.},
with arbitrarily large $\gamma \geq 1$, where some agents can deviate to a remote location and gain, by becoming singleton clusters, while maintaining the desirable stability factor of the declared instance (see also Fig.~\ref{f:singletonClusters}). 

To deal with singleton deviations%
\footnote{Another natural way to deal efficiently with singleton deviations is through some means of \emph{location verification}, such as winner-imposing verification \cite{FT10} or $\eps$-symmetric verification \cite{FZ13,FTZ16}. Adding e.g., winner-imposing verification to the optimal mechanism, discussed in Section~\ref{s:optimal}, results in a strategyproof mechanism for $(2+\sqrt{3})$-stable instances whose optimal clustering may include singleton clusters.},
we should place the facility either at a location close to an extreme one, as we do in Section~\ref{s:rightmost} with the \rightm\ mechanism, or at a random location, as we do in Section~\ref{s:random} with the \rand\ mechanism. More specifically, in Section~\ref{s:rightmost}, we show that the \rightm\ mechanism, which places the facility of any non-singleton optimal cluster at the location of the second rightmost agent, is strategyproof for $5$-stable instances of $k$-Facility Location (even if their optimal clustering includes singleton clusters) and achieves an approximation ratio at most $(n-2)/2$ (Theorem~\ref{thm:rightmost}). Moreover, in Section~\ref{s:random}, we show that the \rand\ mechanism, which places the facility of any optimal cluster at a  location chosen uniformly at random, is strategyproof for $5$-stable instances (again even if their optimal clustering includes singleton clusters) and achieves an approximation ratio of $2$ (Theorem~\ref{thm:random}). 

To obtain a deeper understanding of the challenges behind the design of strategyproof mechanisms for stable instances of $k$-Facility Location on the line, we strengthen the impossibility result of \cite[Theorem~3.7]{FT12} so that it applies to $\gamma$-stable instances with $\gamma < \sqrt{2}$ (Section~\ref{s:lower_bound}). Through a careful analysis of the image sets of deterministic strategyproof mechanisms, we show that for any $k \geq 3$, any $\delta > 0$, and any $\rho \geq 1$, there do not exist any $\rho$-approximate deterministic anonymous strategyproof mechanisms for $(\sqrt{2}-\delta)$-stable instances of $k$-Facility Location on the line (Theorem~\ref{thm:lower_bound}). The proof of Theorem~\ref{thm:lower_bound} requires  additional ideas and extreme care (and some novelty) in the agent deviations, so as to only consider stable instances, compared against the proof of \cite[Theorem~3.7]{FT12}. Interestingly, singleton deviations play a crucial role in the proof of Theorem~\ref{thm:lower_bound}. 

\subsection{Other Related Work}
\label{s:previous}

Approximate mechanism design without money for variants and generalizations of Facility Location games on the line has been a very active and productive area of research in the last decade. 

Previous work has shown that deterministic strategyproof mechanisms can only achieve a bounded approximation ratio for $k$-Facility Location on the line, only if we have at most $2$ facilities \cite{FT12,PT09}. Notably, stable (called \emph{well-separated} in \cite{FT12}) instances with $n = k+1$ agents play a key role in the proof of inapproximability of $k$-Facility Location by deterministic anonymous strategyproof mechanisms \cite[Theorem~3.7]{FT12}. On the other hand, randomized mechanisms are known to achieve a better approximation ratio for $k = 2$ facilities \cite{LWZ09}, a constant approximation ratio if we have $k \geq 2$ facilities and only $n = k+1$ agents \cite{EGTPS11,FT13}, and an approximation ratio of $n$ for any $k \geq 3$ \cite{FT13}. Fotakis and Tzamos \cite{FT10} considered winner-imposing randomized mechanisms that achieve an approximation ratio of $4k$ for $k$-Facility Location in general metric spaces. In fact, the approximation ratio can be improved to $\Theta(\ln k)$, using the analysis of \cite{AV07}.

For the objective of maximum agent cost, Alon et al. \cite{AFPT09} almost completely characterized the approximation ratios achievable by randomized and deterministic strategyproof mechanisms for $1$-Facility Location in general metrics and rings. Fotakis and Tzamos \cite{FT13} presented a $2$-approximate randomized group strategyproof mechanism for $k$-Facility Location on the line and the maximum cost objective. For $1$-Facility Location on the line and the objective of minimizing the sum of squares of the agent connection costs, Feldman and Wilf \cite{FW11} proved that the best approximation ratio is $1.5$ for randomized and $2$ for deterministic mechanisms. 
Golomb and Tzamos \cite{GT17} presented tight (resp. almost tight) additive approximation guarantees for locating a single (resp. multiple) facilities on the line and the objectives of the maximum cost and the social cost. 

Regarding the application of perturbation stability, we follow the approach of \emph{beyond worst-case analysis} (see e.g., \cite{Rough19,Rough20}), where researchers seek a theoretical understanding of the superior practical performance of certain algorithms by formally analyzing them on practically relevant instances.
%
%
The beyond worst-case approach is not anything new for Algorithmic Mechanism Design. \emph{Bayesian} analysis, where the bidder valuations are drawn as independent samples from a distribution known to the mechanism, is standard in revenue maximization when we allocate private goods (see e.g., \cite{RoughOpt}) and has led to many strong and elegant results for social welfare maximization in combinatorial auctions by truthful posted price mechanisms (see e.g., \cite{DFKL17,FGL14}). However, in this work, instead of assuming (similar to Bayesian analysis) that the mechanism designer has a relatively accurate knowledge of the distribution of agent locations on the line (and use e.g., an appropriately optimized percentile mechanism \cite{SuiBS15}), we employ a deterministic restriction on the class of instances (namely, perturbation stability), and investigate if deterministic (resp. randomized) strategyproof mechanisms with a bounded (resp. constant) approximation ratio are possible for locating any number $k \geq 2$ facilities on such instances. To the best of our knowledge, the only previous work where the notion of perturbation stability is applied to Algorithmic Mechanism Design (to combinatorial auctions, in particular) is \cite{FF20} (but see also \cite{BDO18,EFF20} where the similar in spirit assumption of endowed valuations was applied to combinatorial markets).

\section{Notation, Definitions and Preliminaries}
\label{s:prelim}

We let $[n] = \{1, \ldots, n\}$. For any $x, y \in \reals$, we let $d(x, y) = |x - y|$ be the distance of locations $x$ and $y$ on the real line.   
For a tuple $\vec{x} = (x_1, \ldots, x_n) \in \reals^n$, we let $\vec{x}_{-i}$ denote the tuple $\vec{x}$ without coordinate $x_i$. For a non-empty set $S$ of indices, we let $\vec{x}_S = (x_i)_{i \in S}$ and $\vec{x}_{-S} = (x_i)_{i \not\in S}$. We write $(\vec{x}_{-i}, a)$ to denote the tuple $\vec{x}$ with $a$ in place of $x_i$, $(\vec{x}_{-\{i,j\}}, a, b)$ to denote the tuple $\vec{x}$ with $a$ in place of $x_i$ and $b$ in place of $x_j$, and so on.
For a random variable $X$, $\Exp(X)$ denotes the expectation of $X$. For an event $E$ in a sample space, $\Prob(E)$ denotes the probability that $E$ occurs.

\smallskip\noindent\textbf{Instances.}
We consider $k$-Facility Location with $k \geq 2$ facilities and $n \geq k+1$ agents on the real line. We let $N = \{ 1, \ldots, n\}$ be the set of agents. Each agent $i \in N$ resides at a location $x_i \in \reals$, which is $i$'s private information. 
We usually refer to a locations profile $\vec{x} = (x_1, \ldots, x_n) \in \reals^n$, $x_1 \leq \cdots \leq x_n$, as an \emph{instance}. By slightly abusing the notation, we use $x_i$ to refer both to the agent $i$'s location and sometimes to the agent $i$ (i.e., the strategic entity) herself. 

\smallskip\noindent\textbf{Mechanisms.}
A \emph{deterministic mechanism} $M$ for $k$-Facility Location maps an instance $\vec{x}$ to a $k$-tuple $(c_1, \ldots, c_k) \in \reals^k$, $c_1 \leq \cdots \leq c_k$, of facility locations. We let $M(\vec{x})$ denote the outcome of $M$ in instance $\vec{x}$, and let $M_j(\vec{x})$ denote $c_j$, i.e., the $j$-th smallest coordinate in $M(\vec{x})$. We write $c \in M(\vec{x})$ to denote that $M(\vec{x})$ places a facility at location $c$.
A \emph{randomized mechanism} $M$ maps an instance $\vec{x}$ to a probability distribution over $k$-tuples $(c_1, \ldots, c_k) \in \reals^k$.

\smallskip\noindent\textbf{Connection Cost and Social Cost.}
Given a $k$-tuple $\vec{c} = (c_1, \ldots, c_k)$, $c_1 \leq \cdots \leq c_k$, of facility locations, the connection cost of agent $i$ wrt. $\vec{c}$, denoted $d(x_i, \vec{c})$, is
\( d(x_i, \vec{c}) = \min_{1 \leq j \leq k} |x_i - y_j| \).
Given a deterministic mechanism $M$ and an instance $\vec{x}$, $d(x_i, M(\vec{x}))$ denotes the connection cost of agent $i$ wrt. the outcome of $M(\vec{x})$.
If $M$ is a randomized mechanism, the expected connection cost of agent $i$ is
\( \Exp_{\vec{c} \sim M(\vec{x})}(d(x_i, \vec{c})) \).
The \emph{social cost} of a deterministic mechanism $M$ for an instance $\vec{x}$ is
\(\SC(\vec{x}, M(\vec{x})) = \sum_{i = 1}^n d(x_i, M(\vec{x})) \).
The social cost of a facility locations profile $\vec{c} \in \reals^k$ is 
\( \SC(\vec{x}, \vec{c}) = \sum_{i = 1}^n d(x_i, \vec{c}) \).
The \emph{expected social cost} of a randomized mechanism $M$ on instance $\vec{x}$ is
\[ \SC(\vec{x}, M(\vec{x})) = \sum_{i = 1}^n \Exp_{\vec{c} \sim M(\vec{x})}(d(x_i, \vec{c}))  \,.\] 
The \emph{optimal social cost} for an instance $\vec{x}$ is \(  \SC^\ast(\vec{x}) = \min_{\vec{c} \in \reals^k} \sum_{i = 1}^n d(x_i, \vec{c})  \). For $k$-Facility Location, the optimal social cost (and the corresponding optimal facility locations profile) can be computed in $O(kn\log n)$ time by standard dynamic programming.

\smallskip\noindent\textbf{Approximation Ratio.}
A mechanism $M$ has an approximation ratio of $\rho \geq 1$, if for any instance $\vec{x}$, $\SC(\vec{x}, M(\vec{x})) \leq \rho\,\SC^\ast(\vec{x})$. We say that the approximation ratio $\rho$ of $M$ is \emph{bounded}, if $\rho$ is bounded from above either by a constant or by a (computable) function of $n$ and $k$.

\smallskip\noindent\textbf{Strategyproofness.}
A deterministic mechanism $M$ is \emph{strategyproof}, if no agent can benefit from misreporting her location. Formally, $M$ is strategyproof, if for all location profiles $\vec{x}$, any agent $i$, and all locations $y$,
\( d(x_i, M(\vec{x})) \leq d(x_i, M((\vec{x}_{-i}, y))\).
Similarly, a randomized mechanism $M$ is strategyproof (in expectation), if for all location profiles $\vec{x}$, any agent $i$, and all locations $y$,
\( \Exp_{\vec{c} \sim M(\vec{x})}(d(x_i, \vec{c})) \leq \Exp_{\vec{c} \sim M((\vec{x}_{-i}, y)}(d(x_i, \vec{c})) \).

\smallskip\noindent\textbf{Clusterings.}
A \emph{clustering} (or $k$-clustering, if $k$ is not clear from the context) of an instance $\vec{x}$ is any partitioning $\vec{C} = (C_1, \ldots, C_k)$ of $\vec{x}$ into $k$ sets of consecutive agent locations. We index clusters from left to right. I.e., 
$C_1 = \{ x_1, \ldots, x_{|C_1|} \}$, $C_2 = \{ x_{|C_1|+1}, \ldots, x_{|C_1|+|C_2|} \}$, and so on. We refer to a cluster $C_i$ that includes only one agent (i.e., with $|C_i|= 1$) as a \emph{singleton} cluster. We sometimes use $(\vec{x}, \vec{C})$ to highlight that we consider $\vec{C}$ as a clustering of instance $\vec{x}$. 

Two clusters $C$ and $C'$ are identical, denoted $C = C'$, if they include the exact same locations. Two clusterings $\vec{C} = (C_1, \ldots, C_k)$ and $\vec{Y} = (Y_1, \ldots, Y_k)$ of an instance $\vec{x}$ are the same, if $C_i = Y_i$, for all $i \in [k]$. Abusing the notation, we say that a clustering $\vec{C}$ of an instance $\vec{x}$ is identical to a clustering $\vec{Y}$ of a $\gamma$-perturbation $\vec{x}'$ of $\vec{x}$ (see also Definition~\ref{def:stable}), if $|C_i| = |Y_i|$, for all $i \in [k]$. 

We let $x_{i,l}$ and $x_{i, r}$ denote the leftmost and the rightmost agent of each cluster $C_i$. Under this notation, $x_{i-1, r} < x_{i, l} \leq x_{i, r} < x_{i+1, l}$, for all $i \in \{2, \ldots, k-1\}$. Exploiting the linearity of instances, we extend this notation to refer to other agents by their relative location in each cluster. Namely, $x_{i,l+1}$ (resp. $x_{i,r-1}$) is the second agent from the left (resp. right) of cluster $C_i$\,.
The \emph{diameter} of a cluster $C_i$ is $D(C_i) = d(x_{i,l}, x_{i,r})$. The distance of clusters $C_i$ and $C_j$ is $d(C_i, C_j) = \min_{x \in C_i , y\in C_j}\{d(x, y)\}$, i.e., the minimum distance between a location $x \in C_i$ and a location $y \in C_j$. 

A $k$-facility locations (or $k$-centers) profile $\vec{c} = (c_1, \ldots, c_k)$ induces a clustering $\vec{C} = (C_1, \ldots, C_k)$ of an instance $\vec{x}$ by assigning each agent / location $x_j$ to the cluster $C_i$ with facility $c_i$ closest to $x_j$. Formally, for each $i \in [k]$, $C_i = \{ x_j \in \vec{x} : d(x_j, c_i) = d(x_j, \vec{c}) \}$. The \emph{optimal clustering} of an instance $\vec{x}$ is the clustering of $\vec{x}$ induced by the facility locations profile with minimum social cost. 

The social cost of a clustering $\vec{C}$ induced by a $k$-facility locations profile $\vec{c}$ on an instance $\vec{x}$ is simply $\SC(\vec{x}, \vec{c})$, i.e., the social cost of $\vec{c}$ for $\vec{x}$. We sometimes refer to the social cost $\SC(\vec{x}, \vec{C})$ of a clustering $\vec{C}$ for an instance $\vec{x}$, without any explicit reference to the corresponding facility locations profile. Then, we refer to the social cost $\SC(\vec{x}, \vec{c})$, where each facility $c_i$ is located at the median location of $C_i$ (the left median location of $C_i$, if $|C_i|$ is even). 

We often consider certain structural changes in a clustering due to agent deviations. Let $\vec{C}$ be a clustering of an instance $\vec{x}$, which due to an agent deviation, changes to a different clustering $\vec{C}'$. We say that cluster $C_i$ is \emph{split} when $\vec{C}$ changes to $\vec{C}'$, if not all agents in $C_i$ are served by the same facility in $\vec{C}'$. We say that $C_i$ is \emph{merged} in $\vec{C}'$, if all agents in $C_i$ are served by the same facility, but this facility also serves in $\vec{C}'$ some agents not in $C_i$.

\section{Perturbation Stability on the Line: Definition and Properties}
\label{s:stability}

Next, we introduce the notion of $\gamma$-(linear) stability and prove some useful properties of $\gamma$-stable instances of $k$-Facility Location, which are repeatedly used in the analysis of our mechanisms. 

\begin{definition}[$\gamma$-Pertrubation and $\gamma$-Stability]\label{def:stable}
Let $\vec{x} = (x_1, \ldots, x_n)$ be a locations profile. A locations profile $\vec{x}' = (x'_1, \ldots, x'_n)$ is a $\gamma$-perturbation of $\vec{x}$, for some $\gamma \geq 1$, if $x'_1 = x_1$ and for every $i \in [n-1]$, $d(x_i, x_{i+1}) / \gamma \leq d(x'_{i}, x'_{i+1}) \leq d(x_{i}, x_{i+1})$. A $k$-Facility Location instance $\vec{x}$ is $\gamma$-perturbation stable (or simply, $\gamma$-stable), if $\vec{x}$ has a unique optimal clustering $(C_1, \ldots, C_k)$ and every $\gamma$-perturbation $\vec{x}'$ of $\vec{x}$ has the same unique optimal clustering $(C_1, \ldots, C_k)$. 
\end{definition}

Namely, a $\gamma$-perturbation $\vec{x}'$ of an instance $\vec{x}$ is obtained by moving a subset of pairs of consecutive locations closer by a factor at most $\gamma \geq 1$. A $k$-Facility Location instance $\vec{x}$ is $\gamma$-stable, if $\vec{x}$ and any $\gamma$-perturbation $\vec{x}'$ of $\vec{x}$ admit the same unique optimal clustering (where clustering identity for $\vec{x}$ and $\vec{x}'$ is understood as explained in Section~\ref{s:prelim}). 
We consistently select the optimal center $c_i$ of each optimal cluster $C_i$ with an even number of points as the left median point of $C_i$. 

Our notion of linear perturbation stability naturally adapts the notion of metric perturbation stability \cite[Definition~2.5]{AMM17} to the line. 
We note, the class of $\gamma$-stable linear instances, according to Definition~\ref{def:stable}, is at least as large as the class of metric $\gamma$-stable linear instances, according to \cite[Definition~2.5]{AMM17}.
Similarly to \cite[Theorem~3.1]{AMM17} (see also \cite[Lemma~7.1]{Rough17_lect6} and \cite[Corollary~2.3]{ABS12}), we can show that for all $\gamma \geq 1$, every $\gamma$-stable instance $\vec{x}$, which admits an optimal clustering $C_1, \ldots, C_k$ with optimal centers $c_1, \ldots, c_k$, satisfies the following \emph{$\gamma$-center proximity} property: For all cluster pairs $C_i$ and $C_j$, with $i \neq j$, and all locations $x \in C_i$, $d(x, c_j) > \gamma d(x, c_i)$. 

We repeatedly use the following immediate consequence of $\gamma$-center proximity (see also \cite[Lemma~7.2]{Rough17_lect6}). 
The proof generalizes the proof of \cite[Lemma~7.2]{Rough17_lect6} to any $\gamma \geq 2$. 

\begin{proposition}\label{pr:separation}
Let $\gamma \geq 2$ and let $\vec{x}$ be any $\gamma$-stable instance, with unique optimal clustering $C_1, \ldots, C_k$ and optimal centers $c_1, \ldots, c_k$. Then, for all clusters $C_i$ and $C_j$, with $i \neq j$, and all locations $x \in C_i$ and $y \in C_j$, $d(x, y) > (\gamma - 1) d(x, c_i)$. 
\end{proposition}

The following observation, which allows us to treat stability factors multiplicatively, is an immediate consequence of Definition~\ref{def:stable}. 
 
\begin{observation}\label{property_multiplication_observation}
Every $\alpha$-perturbation followed by a $\beta$-perturbation of a locations profile can be implemented by a $(\alpha\beta)$-perturbation and vice versa. Hence, a $\gamma$-stable instance remains $(\gamma/\gamma')$-stable after a $\gamma'$-perturbation, with $\gamma' < \gamma$, is applied to it. 
\end{observation}

We next show that for $\gamma$ large enough, the optimal clusters of a $\gamma$-stable instance are well-separated, in the sense that the distance of two consecutive clusters is larger than their diameters. 

\begin{lemma}[Cluster-Separation Property]\label{l:separated-clusters}
For any $\gamma$-stable instance on the line with optimal clustering $(C_1,\ldots,C_k)$ and all clusters $C_i$ and $C_j$, with $i \neq j$, 
%
%
$d(C_i, C_j)>\frac{(\gamma-1)^2}{2\gamma}\max\{D(C_i),D(C_j)\}$.
\end{lemma}

The cluster-separation property of Lemma~\ref{l:separated-clusters} was first obtained in \cite{ACM0W20} as a consequence of $\gamma$-cluster proximity. For completeness, in Section~\ref{s:app:separated-clusters}, we present a different proof that exploits the linear structure of the instance. Setting $\gamma \geq 2+\sqrt{3}$, we get that: 

\begin{corollary}\label{c:separated-clusters}
Let $\gamma \geq 2+\sqrt{3}$ and let $\vec{x}$ be any $\gamma$-stable instance with unique optimal clustering ($C_1, \ldots, C_k)$. Then, for all clusters $C_i$ and $C_j$, with $i \neq j$, $d(C_i, C_j) > \,\max\{ \Diam(C_i), \Diam(C_j) \}$. 
\end{corollary}

The following is an immediate consequence of the cluster-separation property in Lemma~\ref{l:separated-clusters}.

\begin{observation}\label{obs:cluster_distances}
Let $\vec{x}$ be a $k$-Facility Location with a clustering $\vec{C}=(C_1,\ldots,C_k)$ such that for any two clusters $C_i$ and $C_j$, $\max\{ D(C_i), D(C_j)\} < d(C_i,C_j)$. Then, if in the optimal clustering of $\vec{x}$, there is a facility at the location of some $x \in C_i$, no agent in $C_i$ is served by a facility at $x_j \not\in C_i$.
\end{observation}

Next, we establish the so-called \emph{no direct improvement from singleton deviations} property, used to show the strategyproofness of the \rightm\ and \rand\ mechanisms. Namely, we show that in any $3$-stable instance, no agent deviating to a singleton cluster in the optimal clustering of the resulting instance can improve her connection cost through the facility of that singleton cluster. The proof is deferred to Appendix~\ref{app:s:3_stab_lemma}.

\begin{lemma}\label{3_stab_lemma}
Let $\vec{x}$ be a $\gamma$-stable instance with $\gamma\geq 3$ and optimal clustering $\vec{C}=(C_1,...,C_k)$ and cluster centers $(c_1,...,c_k)$, and let an agent $x_{i}\in C_i \setminus  \{c_i\} $ and a location $x'$ such that $x'$ is a singleton cluster in the optimal clustering of the resulting instance $(\vec{x}_{-i}, x')$. Then, $d(x_i,x') > d(x_{i},c_i)$.
\end{lemma}



The following shows that for $5$-stable instances $\vec{x}$, an agent cannot form a singleton cluster, unless she deviates by a distance larger than the diameter of her cluster in $\vec{x}$'s optimal clustering. 

\begin{lemma}\label{l:singleton_deviation_5}
Let $\vec{x}$ be any $\gamma$-stable instance with $\gamma\geq 5$ and optimal clustering $\vec{C}=(C_1,...,C_k)$. Let $x_{i} \in C_i \setminus  \{c_i\}$ be any agent and $x'$ any location such that $x'$ is a singleton cluster in the optimal clustering of instance $\vec{x}'=(\vec{x}_{-i},x')$, where $x_i$ has deviated to $x'$. Then, $d(x',x_i)>\Diam(C_i)$.
\end{lemma}

\begin{proof}[Sketch.]
Initially, we show that a clustering $\vec{C}'$ of instance $\vec{x}'= (\vec{x}_{-i},x')$, with $d(x',x_i)\leq \Diam(C_i)$, cannot be optimal and contain $x'$ as a singleton cluster, unless some agent $\vec{x}\setminus C_i$ is clustered together with some agent in $C_i$. To this end, we use the lower bound on the distance between difference clusters for $5$-stable instances show in  Lemma~\ref{l:separated-clusters}. Then, using stability arguments, i.e. that the optimal clustering should not change for instance $\vec{x}$, even when we decrease, by a factor of $4$, the distances between consecutive agents in $\vec{x}\setminus C_i$, we show that in $\vec{C}'$ agents in $\vec{x}\setminus C_i$ experience an increase in cost of at least $2\Diam(C_i)$ (notice that $\vec{x}\setminus C_i = \vec{x}'\setminus (C_i \cup \{ x'\})$). But the additional cost of serving $x'$ from $c_i$ in clustering $\vec{C}$ is at most $2\Diam(C_i)$, since $d(x',x_i)\leq\Diam(C_i)$ and $d(x_i, c_i)\leq\Diam(C_i)$. Hence retaining clustering $\vec{C}$ and serving location $x'$ from $c_i$ would have a smaller cost than the supposedly optimal clustering $\vec{C}'$. The complete proof follows by a careful case analysis and can be found in Appendix~\ref{app:s:singleton_deviation_5}.
\qed\end{proof}

\begin{figure}[t]
 \centering\includegraphics[width=0.85\textwidth]{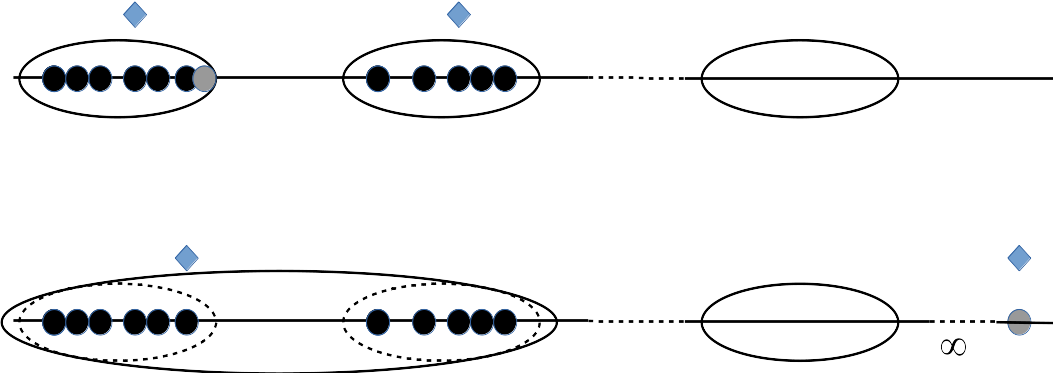}
 \caption{An example of a so-called \emph{singleton deviation}. The deviating agent (grey) declares a remote location, becomes a singleton cluster, and essentially turns the remaining agents into a $(k-1)$-Facility Location instance. Thus, the deviating agent can benefit from her singleton deviation, due to the subsequent cluster merge.}
 \label{f:singletonClusters}
\end{figure}

\begin{algorithm}[tb]
\label{algorithm:optimal}
\DontPrintSemicolon
\SetAlgoLined
\LinesNumbered
\KwResult{An allocation of $k$-facilities}
\KwIn{A $k$-Facility Location instance $\vec{x}$.}
Compute the optimal clustering
 $(C_1, \ldots, C_k)$. Let $c_i$ be the left median point of each cluster $C_i$.\;

 \uIf{\big($\exists i \in [k]$ with $|C_i|=1$\big) or \big($\exists i\in [k-1]$ with $\max\{ D(C_i), D(C_{i+1})\} \geq d(C_i, C_{i+1})$\big)}{
 \KwOut {``FACILITIES ARE NOT ALLOCATED''.}}\uElse{
 
\KwOut {The $k$-facility allocation $(c_1, \ldots, c_k)$ \;}}

\caption{OPTIMAL }
\end{algorithm}

\vspace*{-10mm}
\section{The Optimal Solution is Strategyproof for $(2+\sqrt{3})$-Stable Instances}
\label{s:optimal}

We next show that the \opt\ mechanism, which allocates the facilities optimally, is strategyproof for $(2+\sqrt{3})$-stable instances of $k$-Facility Location whose optimal clustering does not include any singleton clusters. More specifically, in this section, we analyze Mechanism~\ref{algorithm:optimal}.

In general, due to the incentive compatibility of the median location in a single cluster, a deviation can be profitable only if it results in a $k$-clustering different from the optimal clustering $(C_1, \ldots, C_k)$ of $\vec{x}$. For $\gamma$ is sufficiently large, $\gamma$-stability implies that the optimal clusters are well identified 
so that any attempt to alter the optimal clustering (without introducing singleton clusters and without violating the cluster separation property, which is necessary of stability) results in an increased cost for the deviating agent. We should highlight that Mechanism~\ref{algorithm:optimal} may also ``serve'' non-stable instances that satisfy the cluster separation property. We next prove that the mechanism is stategyproof if the true instance is $(2+\sqrt{3})$-stable and its optimal clustering does not include any singleton clusters, when the agent deviations do not introduce any singleton clusters and not result in instances that violate the cluster separation property (i.e. are served by the mechanism) .


\begin{theorem}\label{th:optimal}
The \opt\ mechanism applied to $(2+\sqrt{3})$-stable instances of $k$-Facility Location without singleton clusters in their optimal clustering is strategyproof and minimizes the social cost. 
\end{theorem}


\begin{proof}
We first recall some of the notation about clusterings, introduced in Section~\ref{s:prelim}. Specifically, for a clustering $\vec{C} = (C_1, \ldots, C_k)$ of an instance $\vec{x}$ with centers $\vec{c} = (c_1, \ldots, c_k)$, the cost of an agent (or a location) $x$ is $d(x,\vec{C}) = \min_{j\in[k]}\{d(x, c_j)\}$. The cost of a set of agents $X$ in a clustering $\vec{C}$ is $cost(X,\vec{C})=\sum_{x \in X} d(x_j,\vec{C})$. Finally, the cost of an instance $\vec{x}$ in a clustering $\vec{C}$ is $cost(\vec{x}, \vec{C}) = \sum_{x_j\in \vec{x}} d(x_j,\vec{C})$. This general notation allows us to refer to the cost of the same clustering for different instances. I.e, if $\vec{C}$ is the optimal clustering of $\vec{x}$, then $cost(\vec{y},\vec{C})$ denotes the cost of instance $\vec{y}$ in clustering $\vec{C}$ (where we select the same centers as in clustering $\vec{C}$ for $\vec{x}$). 

The fact that if \opt\ outputs $k$ facilities, they optimize the social cost is straightforward. So, we only need to establish strategyproofness. To this end, we show the following: Let $\vec{x}$ be any $(2+\sqrt{3})$-perturbation stable $k$-Facility Location instance with optimal clustering $\vec{C} = (C_1,\ldots,C_k)$. For any agent $i$ and any location $y$, let $\vec{Y}$ be the optimal clustering of the instance $\vec{y}=(\vec{x}_{-i}, y)$ resulting from the deviation of $i$ from $x_i$ to $y$. Then, if $y$ does not form a singleton cluster in $(\vec{y}, \vec{Y})$, either $d(x_i, \vec{C}) < d(x_i, \vec{Y})$, or there is an $i\in [k-1]$ for which $\max\{ D(Y_i), D(Y_{i+1})\} \geq d(Y_i, Y_{i+1})$. 

So, we let $x_i\in C_i$ deviate to a location $y$, resulting in $\vec{y} = (\vec{x}_{-i}, y)$ with optimal clustering $\vec{Y}$. Since $y$ is not a singleton cluster, it is clustered with agents belonging in one or two clusters of $\vec{C}$, say either in cluster $C_j$ or in clusters $C_{j-1}$ and $C_j$. By optimally of $\vec{C}$ and $\vec{Y}$, the number of facilities serving $C_{j-1} \cup C_{j} \cup \{ y\}$ in $(\vec{y}, \vec{Y})$ is no less than the number of facilities serving $C_{j-1} \cup C_{j}$ in $(\vec{x}, \vec{C})$. Hence, there is at least one facility in either $C_{j-1}$ or $C_{j}$. 

Wlog., suppose that a facility is allocated to an agent in $C_j$ in $(\vec{y},\vec{Y})$. By Corollary~\ref{c:separated-clusters} and Observation~\ref{obs:cluster_distances}, no agent in $C_j$ is served by a facility in $\vec{x} \setminus C_j$ in $\vec{Y}$. Thus we get the following cases about what happens with the optimal clustering $\vec{Y}$ of instance $\vec{y} = (\vec{x}_{-i}, y)$:

\begin{description}
\item[Case 1:] \textit{ $y$ is not allocated a facility in $\vec{Y}$}: This can happen in one of two ways:
    
\begin{description}
        \item[Case 1a:] $y$ is clustered together with some agents from cluster $C_j$ and no facility placed in $C_j$ serves agents in $\vec{x}\setminus C_j$ in $\vec{Y}$. 
        
        \item[Case 1b:] $y$ is clustered together with some agents from a cluster $C_j$ and at least one of the facilities placed in $C_j$ serve agents in $\vec{x}\setminus C_j$ in $\vec{Y}$.
\end{description}
        
\item[Case 2:] \textit{$y$ is allocated a facility in $\vec{Y}$}. This can happen in one of two ways:

 \begin{description}
        \item[Case 2a:] $y$ only serves agents that belong in $C_j$ (by optimality, $y$ must be the median location of the new cluster, which implies that either $y < x_{i,l}$ and $y$ only serves $x_{i,l}$ or $x_{j,l} \leq y \leq x_{j,r}$).
        
        \item[Case 2b:] In $\vec{Y}$, $y$ serves agents that belong in both $C_{j-1}$ and $C_j$.
\end{description}
\end{description}

We next show that the cost of the original clustering $\vec{C}$ is less than the cost of clustering $\vec{Y}$ in $\vec{y}$. Hence, mechanism \opt\ would also select clustering $\vec{C}$ for $\vec{y}$, which would make $x_i$'s deviation to $y$ non-profitable. In particular, it suffices to show that:
\begin{eqnarray*}
cost(\vec{y},\vec{C}) &<& cost (\vec{y},\vec{Y})  \Leftrightarrow\\
cost(\vec{x}, \vec{C}) + d(y,\vec{C}) - d(x_i,\vec{C}) &<& cost (\vec{x},\vec{Y})+ d(y,\vec{Y}) - d(x_i,\vec{Y}) \Leftrightarrow \\
d(y,\vec{C})-d(y,\vec{Y}) &<& cost(\vec{x},\vec{Y}) - cost(\vec{x},\vec{C}) +d(x_i,\vec{C})-d(x_i,\vec{Y})
\end{eqnarray*}
Since $x_i$'s deviation to $y$ is profitable, $d(x_i,\vec{C})-d(x_i,\vec{Y})>0$. Hence, it suffices to show that:
\begin{eqnarray}
d(y,\vec{C})-d(y,\vec{Y}) &\leq& cost(\vec{x},\vec{Y}) - cost(\vec{x},\vec{C}) \nonumber\\
&=& cost(C_j,\vec{Y}) - cost(C_j,\vec{C})  + cost(\vec{x}\setminus C_j,\vec{Y}) - cost(\vec{x}\setminus C_j,\vec{C})\label{eq:pr1-target}
\end{eqnarray}

We first consider Case~1a and Case~2a, i.e., the cases where $\vec{Y}$ allocates facilities to agents of $C_j$ (between $x_{j,l}$ and $x_{j,r}$) that serve only agents in $C_j$. Note that in case 2a, $y$ can also be located outside of $C_j$ and serve only $x_{i,l}$. We can treat this case as Case~1a, since it is equivalent to placing the facility on $x_{i,l}$ and serving $y$ from there. 

In Case~1a and Case~2a, we note that \eqref{eq:pr1-target} holds if the clustering $\vec{Y}$ allocates a single facility to agents in $C_j \cup \{ y \}$, because the facility is allocated to the median of $C_j \cup \{ y \}$, hence $d(y,\vec{C})-d(y,\vec{Y}) = cost(C_j,\vec{Y}) - cost(C_j,\vec{C})$, while $cost(\vec{x}\setminus C_j,\vec{Y}) - cost(\vec{x}\setminus C_j,\vec{C}) \geq 0$, since $\vec{C}$ is optimal for $\vec{x}$. So, we focus on the most interesting case where the agents in $C_j \cup \{ y \}$ are allocated at least two facilities. We observe that \eqref{eq:pr1-target} follows from: 
\begin{align}
d(y,\vec{C}) - d(y,\vec{Y}) & \leq \tfrac{1}{\gamma} \Big( cost(\vec{x}\setminus C_j,\vec{Y}) - cost(\vec{x}\setminus C_j,\vec{C}) \Big)  \label{eq:th1_1} \\
cost(C_j,\vec{C}) - cost(C_j,\vec{Y}) & \leq
\left(1-\tfrac{1}{\gamma}\right) \Big( cost(\vec{x}\setminus C_j,\vec{Y}) - cost(\vec{x}\setminus C_j,\vec{C}) \Big)  
\label{eq:th1_2}
\end{align}

To establish \eqref{eq:th1_1} and \eqref{eq:th1_2}, we first consider the valid $\gamma$-perturbation of the original instance $\vec{x}$ where all distances between consecutive agent pairs to the left of $C_j$ (i.e. agents $\{x_1,x_2,\ldots,x_{j-1,r}\}$) and between consecutive agent pairs to the right of $C_j$ (i.e. agents $\{x_{j+1,l},\ldots,x_{k,r}\}$) are scaled down by $\gamma$. By stability, the clustering $\vec{C}$ remains the unique optimal clustering for the perturbed instance $\vec{x}'$. Moreover, since agents in $\vec{x}\setminus C_j$ are not served by a facility in $C_j$ in $\vec{C}$ and $\vec{Y}$, and since all distances outside $C_j$ are scaled down by $\gamma$, while all distances within $C_j$ remain the same, the cost of the clusterings $\vec{C}$ and $\vec{Y}$ for the perturbed instance $\vec{x}'$ is $cost(C_j, \vec{C}) + cost(\vec{x} \setminus C_j, \vec{C})/\gamma$ and $cost(C_j, \vec{Y}) + cost(\vec{x} \setminus C_j, \vec{Y})/\gamma$, respectively. Using $cost(\vec{x}', \vec{C}) < cost(\vec{x}', \vec{Y})$ and $\gamma \geq 2$, we obtain:
\begin{align}
cost(C_j,\vec{C}) - cost(C_j,\vec{Y}) & < \tfrac{1}{\gamma} \Big( cost(\vec{x}\setminus C_j,\vec{Y}) - cost(\vec{x}\setminus C_j, \vec{C}) \Big) \label{eq:pr1-perturbation-optimality} \\
& \leq \left( 1- \tfrac{1}{\gamma}\right) \Big( cost(\vec{x}\setminus C_j,\vec{Y}) - cost(\vec{x}\setminus C_j, \vec{C}) \Big) \label{eq:pr1-perturbation-final}
\end{align}

Moreover, if $C_j \cup \{ y \}$ is served by at least two facilities in $\vec{Y}$, the facility serving $y$ (and some agents of $C_j$) is placed at the median location of $\vec{Y}$'s cluster that contains $y$. Wlog., we assume that $y$ lies on the left of the median of $C_j$. Then, the decrease in the cost of $y$ due to the additional facility in $\vec{Y}$ is equal to the decrease in the cost of $x_{i,l}$ in $\vec{Y}$, which bounds from below the total decrease in the cost of $C_j$ due to the additional facility in $\vec{Y}$. Hence, %
\begin{equation}\label{eq:pr1-y-improve}
d(y,\vec{C})-d(y,\vec{Y}) \leq cost(C_j,\vec{C}) - cost(C_j,\vec{Y}) 
\end{equation}
We conclude Case~1a and Case~2a, by observing that \eqref{eq:th1_1} follows directly from \eqref{eq:pr1-y-improve} and \eqref{eq:pr1-perturbation-optimality}. 

Finally, we study Case~1b and Case~2b, i.e, the cases where some agents of $C_j$ are clustered with agents of $\vec{x}\setminus C_j$ in $\vec{Y}$. Let $C_{j1}'$ and $C_{j2}'$ denote the clusters of $(\vec{y}, \vec{Y})$ including all agents of $C_j$ (i.e., $C_j \subseteq C_{j1}'\cup C_{j2}'$). By hypothesis, at least one of $C_{j1}'$ and $C_{j2}'$ contains an agent $z \in\vec{x}\setminus C_j$. Suppose this is true for the cluster $C_{j1}'$. Then, $D(C_{j1}') > D(C_j)$, since by Corollary~\ref{c:separated-clusters}, for any $\gamma \geq (2+\sqrt{3})$, the distance of any agent $z$ outside $C_j$ to the nearest agent in $C_j$ is larger than $C_j$'s diameter. But since both $C_{j1}'$ and $C_{j2}'$ contain agents of $C_j$, we have that $d(C_{j1}',C_{j2}') < D(C_j)$. Therefore, $D(C_{j1}')> d(C_{j1}',C_{j2}')$ and the cluster-separation property is violated. Hence the resulting instance $\vec{y}$ is not $\gamma$-stable and Mechanism~\ref{algorithm:optimal} does not allocated any facilities for it. 
\qed\end{proof}

\section{A Deterministic Mechanism Resistant to Singleton Deviations}
\label{s:rightmost}

Next, we present a deterministic strategyproof mechanism for $5$-stable instances of $k$-Facility Location whose optimal clustering may include singleton clusters. To make  singleton cluster deviations non profitable, cluster merging has to be discouraged by the facility allocation rule. So, we allocate facilities near the edge of each optimal cluster, ending up with a significantly larger approximation ratio and a requirement for larger stability, in order to achieve strategyproofness. Specifically, we now need to ensure that no agent can become a singleton cluster close enough to her original location. Moreover, since agents can now gain by splitting their (true) optimal cluster, we need to ensure that such deviations are either non profitable or violate the cluster-separation property. 

\begin{algorithm}[b]
\label{algorithm:deterministic}
\DontPrintSemicolon
\SetAlgoLined
\LinesNumbered
\KwResult{An allocation of $k$-facilities}
\KwIn{A $k$-Facility Location instance $\vec{x}$.}
Find the optimal clustering $\vec{C}=(C_1,\ldots,C_k)$ of $\vec{x}$. \;

 \uIf{there are two consecutive clusters $C_i$ and $C_{i+1}$ with $\max\{ D(C_i), D(C_{i+1})\} \geq d(C_i,C_{i+1})$}{ 
 \KwOut {``FACILITIES ARE NOT ALLOCATED''.}}

\For{$i\in \{1,\ldots,k\}$}{
\uIf{$|C_i|>1$} {Allocate a facility to the location of the second rightmost agent of $C_i$, i.e., $c_i \leftarrow x_{i,r-1}$. }
\Else{
Allocate a facility to the single agent location of $C_i$: $c_i \leftarrow x_{i,l}$}
}
\KwOut{The $k$-facility allocation $\vec{c}=(c_1,\ldots,c_k)$.}
\caption{\rightm }
\end{algorithm}

\begin{theorem}\label{thm:rightmost}
\rightm\ (Mechanism~\ref{algorithm:deterministic}) is strategyproof for $5$-stable instances of $k$-Facility Location and achieves an approximation ratio of $(n-2)/2$.
\end{theorem}

\begin{proof}
The approximation ratio of $(n-2)/2$ follows directly from the fact that the mechanism allocates the facility to the second rightmost agent of each non-singleton optimal cluster.

\vspace*{-0.01mm}As for strategyproofness, let $\vec{x}$ denote the true instance and $\vec{C} = (C_1,\ldots,C_k)$ its optimal clustering. We consider an agent $x_i\in C_j$ deviating to location $y$, resulting in an instance $\vec{y} =(\vec{x}_{-i}, y)$ with optimal clustering $\vec{Y}$. Agent $x_i$'s cost is at most $\Diam(C_j)$. Agent $x_{i}$ could profitably declare false location $y$ in the following ways:
\begin{description}
    \item[Case 1:] The agents in $C_j$ are clustered together in $\vec{Y}$ and $y$ is allocated a facility with $d(y,x_{i}) < d(x_{i}, x_{i,r-1}) \leq \Diam(C_j)$ ($x_{i,r-1}$ is the location of $x_{i}$'s facility, when she is truthful).
    
\begin{description}
    \item[Case 1a:] $y$ is a singleton cluster and $d(y,x_{i}) < \Diam(C_j)$. For $5$-stable instances, Lemma~\ref{l:singleton_deviation_5} implies that $x_i \in C_j$ has to move by at least $\Diam(C_j)$ to become a singleton cluster, a contradiction.
    
   \item[Case 1b:] $y$ is the second rightmost agent of a cluster $C_{j}'$ in $(\vec{y}, \vec{Y})$. Then, the agent $x_i$ can gain only if $d(y, x_{i}) < \Diam(C_j)$. In Case~1, the agents in $C_j$ are clustered together in $\vec{Y}$. If $y < x_i$, $y$ must be the second rightmost agent of a cluster on the left of $x_{j,l}$ and by  Lemma~\ref{l:separated-clusters}, $d(x_i, y) \geq d(x_{j,l}, x_{j-1,r}) > \Diam(C_j)$. Hence, such a deviation cannot be profitable for $x_i$ (note how this case crucially uses the facility allocation to the second rightmost agent of a cluster). If $y > x_i$, $x_i$ can only gain if $y$ is the second rightmost agent of a cluster including $C_j \cup \{ y, x_{j+1,l} \}$ and possibly some agents on the left of $C_j$, which is treated below.
\end{description}

    \item[Case 2:] The agents in $C_j$ are clustered together in $\vec{Y}$ and $C_j$ is merged with some agents from $C_{j+1}$ and possibly some other agents to the left of $x_{j,l}$ (note that merging $C_j$ only with agents to the left of $x_{j,l}$ does not change the facility of $x_i$). Then, we only need to consider the case where the deviating agent $x_i$ is $x_{j,r}$, since any other agent to the left of $x_{j,r-1}$ cannot gain, because cluster merging can only move their serving facility further to the right. As for $x_{j,r}$, we note that by optimality and the hypothesis that agents in $C_j$ belong in the same cluster of $\vec{Y}$, $x_{i,r}$ cannot cause the clusters $C_j$ and $C_{j+1}$ to merge in $\vec{Y}$ by deviating in the range $[x_{j,r},x_{j+1,l}]$. The reason is that the set of agents $(C_i\setminus \{x_{j,r}\}) \cup \{ y \} \cup C_{j+1}$ cannot be served optimally by a single facility, when the set of agents $C_j \cup C_{j+1}$ requires two facilities in the optimal clustering $\vec{C}$. Hence, unless $C_{j+1}$ is split in $\vec{Y}$ (which is treated similarly to Case~3a), $x_{j,r}$ can only move her facility to $C_{j+1}$, which is not profitable for her, due to Lemma~\ref{l:separated-clusters}. 

     \item[Case 3:] $C_j$ is split into two clusters in $\vec{Y}$. Hence, the leftmost agents, originally in $C_j$, are served by a different facility than the rest of the agents originally in $C_j$. We next show that in any profitable deviation of $x_{i}$ where $C_j$ is split, either the deviation is not feasible or the cluster-separation property is violated. The case analysis below is similar to the proof of Theorem~\ref{th:optimal}. 

\begin{description}
\item[Case 3a:] Agents in $C_j$ are clustered together with some agents of $\vec{x}\setminus C_j$ in $\vec{Y}$. By hypothesis, there are agents $z, w \in C_j$ placed in different clusters of $\vec{Y}$, and at least one of them, say $z$, is clustered together with an agent $p \in C_\ell$, with $\ell \neq j$, in $\vec{Y}$. For brevity, we refer to the (different) clusters in which $z$ and $w$ are placed in clustering $\vec{Y}$ as $C_z'$ and $C_w'$, respectively. Then, $\Diam(C_z') \geq d(p,z) > \Diam(C_j)$, by Lemma~\ref{l:separated-clusters}. But also $d(C_z', C_w') < d(z,w) \leq \Diam(C_j)$, and consequently, $\Diam(C_z') > d(C_z', C_w')$, which implies that the cluster-separation property is violated and Mechanism~\ref{algorithm:deterministic} does not allocate any facilities in this case. 

\item[Case 3b:] Agents in $C_j$ are split and are not clustered together with any agents of $\vec{x}\setminus C_j$ in $\vec{Y}$. Hence, $y$ is not clustered with any agents in $\vec{x}\setminus C_j$ in $\vec{Y}$. Otherwise, i.e., if $y$ is not clustered with agents of $C_j$ in $\vec{Y}$, it would be suboptimal for clustering $\vec{Y}$ to allocate more than one facility to agents of $C_j \setminus \{ x_i \}$ and at most $k-2$ facilities to $(\vec{x} \cup \{ y \}) \setminus C_j$, while the optimal clustering $\vec{C}$ allocates a single facility to $C_j$ and $k-1$ facilities to $\vec{x} \setminus C_j$. But again if $y$ is only clustered with agents of $C_j$, it is suboptimal for clustering $\vec{Y}$ to allocate more than one facility to agents of $(C_j \cup \{ y \}) \setminus \{ x_i \}$ and at most $k-2$ facilities to $\vec{x} \setminus C_j$, while the optimal clustering $\vec{C}$ allocates a single facility to $C_j$ and $k-1$ facilities to $\vec{x} \setminus C_j$, as shown in the proof of Theorem~\ref{th:optimal}.\qed
\end{description}\end{description}
\end{proof}

\section{Low Stability and Inapproximability by Deterministic Mechanisms}
\label{s:lower_bound}

We next extend the impossibility result of \cite[Theorem~3.7]{FT12} to $\sqrt{2}$-stable instances of $k$-Facility Location on the line, with $k \geq 3$. Thus, we provide strong evidence that restricting our attention to stable instances does not make strategyproof mechanism design trivial. 

\subsection{Image Sets, Holes and Well-Separated Instances}

We start with some basic facts about strategyproof mechanisms and by adapting the technical machinery of well-separating instances from \cite[Section~2.2]{FT12} to stable instances. 

\smallskip\noindent{\bf Image Sets and Holes.}
Given a mechanism $M$, the \emph{image set} $I_i(\vec{x}_{-i})$ of an agent $i$ with respect to an instance $\vec{x}_{-i}$ is the set of facility
locations the agent $i$ can obtain by varying her reported location. Formally,
\( I_i(\vec{x}_{-i}) = \{ a \in \reals: \exists y \in \reals \mbox{ with } M(\vec{x}_{-i}, y) = a \} \).

If $M$ is strategyproof, any image set $I_i(\vec{x}_{-i})$ is a collection of closed intervals (see e.g., \cite[p.~249]{SV07}). Moreover, a strategyproof mechanism $M$ places a facility at the location in $I_i(\vec{x}_{-i})$ nearest to the declared location of agent $i$. Formally, for any agent $i$, all instances $\vec{x}$, and all locations $y$, 
\( d(y, M(\vec{x}_{-i}, y)) =
\inf_{a \in I_i(\vec{x}_{-i})}\{d(y, a)\} \).

Some care is due, because we consider mechanisms that need to be strategyproof only for $\gamma$-stable instances $(\vec{x}_{-i}, y)$. The image set of such a mechanism $M$ is well defined (possibly by assuming that all facilities are placed to essentially $+\infty$), whenever $(\vec{x}_{-i}, y)$ is not $\gamma$-stable. Moreover, the requirement that $M$ places a facility at the location in $I_i(\vec{x}_{-i})$ nearest to the declared location $y$ of agent $i$ holds only if the resulting instance $(\vec{x}_{-i}, y)$ is stable. We should underline that all instances considered in the proof of Theorem~\ref{thm:lower_bound} are stable (and the same holds for the proofs of the propositions adapted from \cite[Section~2.2]{FT12}). 

Any (open) interval in the complement of an image set $I \equiv I_i(\vec{x}_{-i})$ is called a \emph{hole} of $I$. Given a location $y \not\in I$, we let $l_y = \sup_{a \in I} \{ a < y\}$ and $r_y = \inf_{a \in I} \{ a > y\}$ be the locations in $I$ nearest to $y$ on the left and on the right, respectively. Since $I$ is a collection of closed intervals, $l_y$ and $r_y$ are well-defined and satisfy $l_y < y < r_y$. For convenience, given a $y \not\in I$, we refer to the interval $(l_y, r_y)$ as a $y$-hole in $I$.

\smallskip\noindent{\bf Well-Separated Instances.}
Given a deterministic strategyproof mechanism $M$ with a bounded approximation $\rho \geq 1$ for $k$-Facility Location, an instance $\vec{x}$ is $(x_1|\cdots|x_{k-1}|x_k,x_{k+1})$-\emph{well-separated} if $x_1 < \cdots < x_k < x_{k+1}$ and 
\( \rho d(x_{k+1}, x_{k}) < \min_{i \in \{2, \ldots, k\}} \{ d(x_{i-1},  x_i) \} \).
We call $x_k$ and $x_{k+1}$ the \emph{isolated pair} of the well-separated instance $\vec{x}$.

Hence, given a $\rho$-approximate mechanism $M$ for $k$-Facility Location, a well-separated instance includes a pair of nearby agents at distance to each other less than $1/\rho$ times the distance between any other pair of consecutive agents. Therefore, any $\rho$-approximate mechanism serves the two nearby agents by the same facility and serve each of the remaining ``isolated'' agents by a different facility. We remark that well-separated instances are also $\rho$-stable. 

We are now ready to adapt some useful properties of well-separated instances from \cite[Section~2.2]{FT12}. It is not hard to verify that the proofs of the auxiliary lemmas below apply to $\sqrt{2}$-stable instances, either without any change or with some minor modifications (see also \cite[Appendix~A]{FT12}). For completeness, we give the proofs of the lemmas below in Appendix~\ref{s:app:well-separated}. 

\begin{lemma}[Proposition~2.2, \cite{FT12}]\label{l:nice-allocation}
Let $M$ be any deterministic startegyproof mechanism with a bounded approximation ratio $\rho \geq 1$. For any $(x_1|\cdots|x_{k-1}|x_k,x_{k+1})$-well-separated instance $\vec{x}$, $M_k(\vec{x}) \in [x_k, x_{k+1}]$.
\end{lemma}

\begin{lemma}[Proposition 2.3, \cite{FT12}]\label{l:movingRight}
Let $M$ be any deterministic startegyproof mechanism with a bounded approximation ratio $\rho \geq 1$, and let $\vec{x}$ be a $(x_1|\cdots|x_{k-1}|x_k,x_{k+1})$-well-separated instance with $M_k(\vec{x}) = x_k$. Then, for every $(x_1|...|x_{k-1}|x'_k,x'_{k+1})$-well-separated instance $\vec{x}'$ with $x'_k \geq x_{k}$, $M_k(\vec{x}') = x'_k$.
\end{lemma}

\begin{lemma}[Proposition 2.4, \cite{FT12}]\label{l:movingLeft}
Let $M$ be any deterministic startegyproof mechanism with a bounded approximation ratio $\rho \geq 1$, and let $\vec{x}$ be a $(x_1|\cdots|x_{k-1}|x_k,x_{k+1})$-well-separated instance with $M_k(\vec{x}) = x_{k+1}$. Then, for every $(x_1|...|x_{k-1}|x'_k,x'_{k+1})$-well-separated instance $\vec{x}'$ with $x'_{k+1} \leq x_{k+1}$, $M_k(\vec{x}') = x'_{k+1}$.
\end{lemma}

\subsection{The Proof of the Impossibility Result}

We are now ready to establish the main result of this section. The proof of the following builds on the proof of \cite[Theorem~3.7]{FT12}. However, we need some additional ideas and to be way more careful with the agent deviations used in the proof, since our proof can only rely on $\sqrt{2}$-stable instances. 

\begin{theorem}\label{thm:lower_bound}
For every $k \geq 3$ and any $\delta > 0$, any deterministic anonymous strategyproof mechanism for $(\sqrt{2}-\delta)$-stable instances of $k$-Facility Location on the real line with $n \geq k+1$ agents has an unbounded approximation ratio. 
\end{theorem}

\begin{proof}
We only consider the case where $k = 3$ and $n = 4$. It is not hard to verify that the proof applies to any $k \geq 3$ and $n \geq k+1$. To reach a contradiction, let $M$ be any deterministic anonymous strategyproof mechanism for $(\sqrt{2}-\delta)$-stable instances of $3$-Facility Location with $n = 4$ agents and with an approximation ratio of $\rho \geq 1$. 

We consider a $(x_1|x_2|x_3,x_4)$-well-separated instance $\vec{x}$. For a large enough $\lambda \gg \rho$ and a very large (practically infinite) $B \gg 6\rho\lambda$, we let $\vec{x} = (0, \lambda, 6B+\lambda, 6B+\lambda+\eps)$, for some small enough $\eps > 0$ ($\eps \ll \lambda/\rho$). By choosing $\lambda$ and $\eps$ appropriately, becomes the instance $\vec{x}$ $\gamma$-stable, for $\gamma \gg \sqrt{2}$. 

By Lemma~\ref{l:nice-allocation}, $M_3(\vec{x}) \in [ x_3, x_4 ]$. Wlog, we assume that $M_3(\vec{x}) \neq x_3$ (the case where $M_3(\vec{x}) \neq x_4$ is fully symmetric and requires Lemma~\ref{l:movingRight}). Then, by moving agent $4$ to $M_3(\vec{x})$, which results in a well-separated instance and, by strategyproofness, requires that $M$ keeps a facility there, we can assume wlog. that $M_3(\vec{x}) = x_4$. 

Since $\vec{x}$ is well-separated and $M$ is $\rho$-approximate, both $x_3$ and $x_4$ are served by the facility at $x_4$. Hence, there is a $x_3$-hole $h = (l, r)$ in the image set $I_3(\vec{x}_{-3})$. Since $M(\vec{x})$ places a facility at $x_4$ and not in $x_3$, the right endpoint $r$ of $h$ lies between $x_3$ and $x_4$, i.e. $r \in (x_3, x_4]$. Moreover, since $M$ is $\rho$-approximate and strategyproof for $(\sqrt{2}-\delta)$-stable instances, agent $3$ should be served by a facility at distance at most $\rho\lambda$ to her, if she is located at $4B$. Hence, the left endpoint of the hole $h$ is $l > 3B$. We distinguish two cases based on the distance of the left endpoint $l$ of $h$ to $x_4$. 

\noindent\textbf{Case 1: $x_4 - l > \sqrt{2}\lambda$.}
We consider the instance $\vec{y} = (\vec{x}_{-3}, a)$, where $a > l$ is arbitrarily close to $l$ (i.e., $a \gtrsim l$) so that $d(a, x_4) = \sqrt{2}\lambda$. Since $d(x_1, x_2) = \lambda$, $d(x_2, a)$ is quite large, and $d(a, x_4) = \sqrt{2}\lambda$, the instance $\vec{y}$ is $(\sqrt{2}-\delta)$-stable, for any $\delta > 0$. By strategyproofness, $M(\vec{y})$ must place a facility at $l$, since $l \in I_3(\vec{x}_{-3})$.

Now, we consider the instance $\vec{y}' = (\vec{y}_{-4}, l)$. Since we can choose $a > l$ so that $d(l, a) \ll \lambda$, the instance $\vec{y}'$ is $(x_1 | x_2 | l, a)$-well-separated and $(\sqrt{2}-\delta)$-stable. Hence, by strategyproofness, $M(\vec{y}')$ must keep a facility at $l$, because $l \in I_4(\vec{y}_{-4})$. 

Then, by Lemma~\ref{l:movingLeft}, $y'_4 = a \in M(\vec{y}')$, because for  
the $(x_1 | x_2 | x_3, x_4)$-well-separated instance $\vec{x}$,  $M_3(\vec{x}) = x_4$, and $\vec{y}'$ is a $(x_1 | x_2 | l, a)$-well-separated instance with $y'_4 \leq x_4$. Since both $l, a \in M(\vec{y}')$, either agents $1$ and $2$ are served by the same facility of $M(\vec{y}')$ or agent $2$ is served by the facility at $l$. In both cases, the social cost of $M(\vec{y}')$ becomes arbitrarily larger than $a-l$, which is the optimal social cost of the $3$-Facility Location instance $\vec{y}'$. 

\noindent\textbf{Case 2: $x_4 - l \leq \sqrt{2}\lambda$.}
This case is similar to Case~1, but it requires a bit more careful further case analysis. The details can be found in Appendix~\ref{s:app:lower_bound}. 
\qed\end{proof}

\begin{algorithm}[t]
\DontPrintSemicolon
\SetAlgoLined
\LinesNumbered
\KwResult{An allocation of $k$-facilities}
\KwIn{A $k$-Facility Location instance $\vec{x}$.}
Find the optimal clustering $\vec{C}=(C_1,\ldots,C_k)$ of $\vec{x}$. \;

 \uIf{there are two consecutive clusters $C_i$ and $C_{i+1}$ with $1.6\cdot \max\{ D(C_i), D(C_{i+1})\} \geq  d(C_i,C_{i+1}) $}{ 
 \KwOut {``FACILITIES ARE NOT ALLOCATED''.}}

\For{$i\in \{1,\ldots,k\}$}{
Allocate the facility to an agent $c_i$ selected uniformly at random from the agents of cluster $C_i$}
\KwOut{The $k$-facility allocation $\vec{c}=(c_1,\ldots,c_k)$.}
\caption{\rand }
\label{algorithm:randomized}
\end{algorithm}

\section{A Randomized Mechanism with Constant Approximation}
\label{s:random}

In this section, we show that for an appropriate stability, a simple randomized mechanism is strategyproof, can deal with singleton clusters and achieves an approximation ratio of $2$.

The intuition is that the \rightm\ mechanism can be easily transformed to a randomized mechanism, using the same key properties to guarantee strategyproofness, but achieving an $O(1)$-approximation, as opposed to $O(n)$-approximation of \rightm. Specifically, \rand\ (see also Mechanism~\ref{algorithm:randomized}) again finds the optimal clusters, but then places a facility at the location of an agent selected uniformly at random from each optimal cluster. We again use cluster-separation property, as a necessary condition for stability of the optimal clustering. The stability properties required to guarantee strategyproofness are very similar to those required by \rightm, because the set of possible profitable deviations is very similar for \rightm\ and \rand.  Finally, notice that the cluster-separation property step of \rand\ (step 2) now makes use that due to Lemma~\ref{l:separated-clusters}, it must be $1.6\cdot \max\{ D(C_i), D(C_{i+1})\} <  d(C_i,C_{i+1})$ for $5$-stable instances.




\begin{theorem}\label{thm:random}
\rand\ (Mechanism~\ref{algorithm:randomized}) is strategyproof and  achieves an approximation ratio of $2$ for $5$-stable instances of $k$-Facility Location on the line. 
\end{theorem}

\begin{proof}[Sketch.] We present here the outline of the proof. The full proof can be found in Appendix~\ref{app:thm:random}. 
The approximation guarantee is straightforward to verify. As mentioned, the proof of strategyproofness is smilar to the proof of Theorem~\ref{thm:rightmost}. In general, we need to cover the key deviation cases, which include the following:

\begin{description}
\item[Case 1:] why agent deviating agent $x\in C_i$ cannot gain by becoming a member of another cluster, 
\item[Case 2:] or by becoming a self serving center, 

\item[Case 3:] or by merging or splitting $C_i$.
\end{description}

Cases 2 and 3 can be immediately derived from the proof of Theorem~\ref{thm:rightmost}.

The most interesting case is Case~1: $x_i$ deviates to $x'$ to be clustered together with agents from a different cluster of $\vec{C}$, in order to gain, without splitting $C_i$ (again consider $\vec{C}=(C_1,..., C_k)$ the optimal clustering of original instance $\vec{x}$ and $\vec{C}'$ the optimal clustering of instance $\vec{x}'=(\vec{x}_{-i},x')$).

By analyzing the expected value of agent $x_i$ in both clusterings $\vec{C}$ and $\vec{C}'$ we show that in order for her to be able to gain from such a deviation, it must be $d(x',x_i)<\Diam(C_i)$ and $x'$ is clustered together with agents in $C_{i-1}$ or $C_{i+1}$, suppose $C_{i-1}$ w.l.o.g. Since agents in $C_i\setminus x_i$ are not split in clustering $\vec{C}'$, we know they form cluster $C_{i'}'\in \vec{C}'$. Hence, in this case $x\in C_{i'-1}'$. The key to the proof is to show that since $d(x',x_i)<\Diam(C_i)$ then clustering $\vec{C}'$ on instance $\vec{x}'$ violates the cluster separation property verification step, either between clusters $C_{i'}'$ and $C_{i'-1}'$ or between clusters $C_{i'-1}'$ and $C_{i'-2}'$. This is also the reason why in this case the cluster separation property verification step needs to be more precise, for $5$-stable instances, as mentioned in the description of the algorithm.
\qed\end{proof}

%

\bibliographystyle{splncs04}
\bibliography{mechdesign}

\begin{thebibliography}{10}
\providecommand{\url}[1]{\texttt{#1}}
\providecommand{\urlprefix}{URL }
\providecommand{\doi}[1]{https://doi.org/#1}

\bibitem{ACM0W20}
Agarwal, P., Chang, H., Munagala, K., Taylor, E., Welzl, E.: Clustering under
  perturbation stability in near-linear time. In: Proc. of the 40th {IARCS}
  Conference on Foundations of Software Technology and Theoretical Computer
  Science ({FSTTCS} 2020). LIPIcs, vol.~182, pp. 8:1--8:16 (2020)

\bibitem{AFPT09}
Alon, N., Feldman, M., Procaccia, A., Tennenholtz, M.: Strategyproof
  approximation of the minimax on networks. Mathematics of Operations Research
  \textbf{35}(3),  513--526 (2010)

\bibitem{AMM17}
Angelidakis, H., Makarychev, K., Makarychev, Y.: Algorithms for stable and
  perturbation-resilient problems. In: Proc. of the 49th {ACM} Symposium on
  Theory of Computing ({STOC} 2017). pp. 438--451 (2017)

\bibitem{AK08}
Archer, A., Kleinberg, R.: Truthful germs are contagious: {A} local-to-global
  characterization of truthfulness. In: Proc. of the 9th {ACM} Conference on
  Electronic Commerce (EC~'08). pp. 21--30 (2008)

\bibitem{AV07}
Arthur, D., Vassilvitskii, S.: k-means++: the advantages of careful seeding.
  In: Proceedings of the Eighteenth Annual {ACM-SIAM} Symposium on Discrete
  Algorithms ({SODA} 2007). pp. 1027--1035. {SIAM} (2007)

\bibitem{ADPP09}
Auletta, V., Prisco, R.D., Penna, P., Persiano, G.: The power of verification
  for one-parameter agents. Journal of Computer and System Sciences
  \textbf{75},  190--211 (2009)

\bibitem{ABS12}
Awasthi, P., Blum, A., Sheffet, O.: Center-based clustering under perturbation
  stability. Inf. Process. Lett.  \textbf{112}(1-2),  49--54 (2012)

\bibitem{BDO18}
Babaioff, M., Dobzinski, S., Oren, S.: Combinatorial auctions with endowment
  effect. In: Proc. of the 2018 {ACM} Conference on Economics and Computation
  ({EC}~2018). pp. 73--90 (2018)

\bibitem{BHW16}
Balcan, M., Haghtalab, N., White, C.: {$k$-Center Clustering Under Perturbation
  Resilience}. In: Proc. of the 43rd International Colloquium on Automata,
  Languages and Programming ({ICALP}~2016). LIPIcs, vol.~55, pp. 68:1--68:14
  (2016)

\bibitem{BBG13}
Balcan, M.F., Blum, A., Gupta, A.: Clustering under approximation stability.
  Journal of the ACM  \textbf{60}(2) (2013)

\bibitem{BL16}
Balcan, M., Liang, Y.: Clustering under perturbation resilience. {SIAM} Journal
  on Computing  \textbf{45}(1),  102--155 (2016)

\bibitem{BL10}
Bilu, Y., Linial, N.: {Are Stable Instances Easy?} In: Proc. of the 1st
  Symposium on Innovations in Computer Science ({ICS}~2010). pp. 332--341.
  Tsinghua University Press (2010)

\bibitem{BDLS13}
Bilu, Y., Daniely, A., Linial, N., Saks, M.E.: On the practically interesting
  instances of {MAXCUT}. In: Portier, N., Wilke, T. (eds.) Proceedings of the
  30th International Symposium on Theoretical Aspects of Computer Science
  ({STACS} 2013). LIPIcs, vol.~20, pp. 526--537. Schloss Dagstuhl -
  Leibniz-Zentrum f{\"{u}}r Informatik (2013)

\bibitem{CESY12}
Caragiannis, I., Elkind, E., Szegedy, M., Yu, L.: Mechanism design: from
  partial to probabilistic verification. In: Proc. of the 13th {ACM} Conference
  on Electronic Commerce (EC~'12). pp. 266--283 (2012)

\bibitem{Carroll13}
Carroll, G.: When are local incentive constraints sufficient? Econometrica
  \textbf{80}(2),  661--686 (2012)

\bibitem{ChenFLWYZ20}
Chen, Z., Fong, K.C., Li, M., Wang, K., Yuan, H., Zhang, Y.: Facility location
  games with optional preference. Theoretical Computer Science  \textbf{847},
  185--197 (2020)

\bibitem{DFMN12}
Dokow, E., Feldman, M., Meir, R., Nehama, I.: Mechanism design on discrete
  lines and cycles. In: Proc. of the 13th {ACM} Conference on Electronic
  Commerce (EC~'12). pp. 423--440 (2012)

\bibitem{DFKL17}
D{\"u}etting, P., Feldman, M., Kesselheim, T., Lucier, B.: {Prophet
  Inequalities Made Easy: Stochastic Optimization by Pricing Non-Stochastic
  Inputs}. In: Proc. of the 58th Symposium on Foundations of Computer Science
  ({FOCS}~2017). pp. 540--551 (2017)

\bibitem{EGTPS11}
Escoffier, B., Gourv\`{e}s, L., Thang, N., Pascual, F., Spanjaard, O.:
  Strategy-proof mechanisms for {Facility Location} games with many facilities.
  In: Proc. of the 2nd International Conference on Algorithmic Decision Theory
  (ADT~'11). LNAI, vol.~6992, pp. 67--81 (2011)

\bibitem{EFF20}
Ezra, T., Feldman, M., Friedler, O.: A general framework for endowment effects
  in combinatorial markets. In: Proc. of the 2020 {ACM} Conference on Economics
  and Computation ({EC}~2020) (2020)

\bibitem{FLSWZ20}
Feigenbaum, I., Li, M., Sethuraman, J., Wang, F., Zou, S.: Strategic facility
  location problems with linear single-dipped and single-peaked preferences.
  Autonomous Agents and Multi-Agent Systems  \textbf{34}(2), ~49 (2020)

\bibitem{FGL14}
Feldman, M., Gravin, N., Lucier, B.: {Combinatorial Auctions via Posted
  Prices}. In: Proc. of the 26th {ACM-SIAM} Symposium on Discrete Algorithms.
  pp. 123--135 (2014)

\bibitem{FW11}
Feldman, M., Wilf, Y.: Randomized strategyproof mechanisms for {Facility
  Location} and the mini-sum-of-squares objective. CoRR abs  \textbf{1108.1762}
  (2011)

\bibitem{FF20}
Fikioris, G., Fotakis, D.: Mechanism design for perturbation stable
  combinatorial auctions. In: Proceedings of the 13th International Symposium
  on Algorithmic Game Theory ({SAGT} 2020). Lecture Notes in Computer Science,
  vol. 12283, pp. 47--63. Springer (2020)

\bibitem{FM21}
Filimonov, A., Meir, R.: Strategyproof facility location mechanisms on discrete
  trees. CoRR  \textbf{abs/2102.02610} (2021),
  \url{https://arxiv.org/abs/2102.02610}

\bibitem{FT10}
Fotakis, D., Tzamos, C.: Winner-imposing strategyproof mechanisms for multiple
  facility location games. Theoretical Computer Science  \textbf{472},  90--103
  (2013)

\bibitem{FT12}
Fotakis, D., Tzamos, C.: On the power of deterministic mechanisms for facility
  location games. {ACM} Transactions on Economics and Computation
  \textbf{2}(4),  15:1--15:37 (2014)

\bibitem{FT13}
Fotakis, D., Tzamos, C.: Strategyproof facility location for concave cost
  functions. Algorithmica  \textbf{76}(1),  143--167 (2016)

\bibitem{FTZ16}
Fotakis, D., Tzamos, C., Zampetakis, M.: Mechanism design with selective
  verification. In: Proceedings of the 2016 {ACM} Conference on Economics and
  Computation ({EC} 2016). pp. 771--788. {ACM} (2016)

\bibitem{FZ13}
Fotakis, D., Zampetakis, E.: Truthfulness flooded domains and the power of
  verification for mechanism design. {ACM} Transactions on Economics and
  Computation  \textbf{3}(4),  20:1--20:29 (2015)

\bibitem{GH20}
Goel, S., Hann{-}Caruthers, W.: Coordinate-wise median: Not bad, not bad,
  pretty good. CoRR  \textbf{abs/2007.00903} (2020),
  \url{https://arxiv.org/abs/2007.00903}

\bibitem{GT17}
Golomb, I., Tzamos, C.: Truthful facility location with additive errors. CoRR
  \textbf{abs/1701.00529} (2017), \url{http://arxiv.org/abs/1701.00529}

\bibitem{GL86}
Green, J., Laffont, J.: Partially verifiable information and mechanism design.
  Review of Economic Studies  \textbf{53}(3),  447--456 (1986)

\bibitem{KVZ19}
Kyropoulou, M., Ventre, C., Zhang, X.: Mechanism design for constrained
  heterogeneous facility location. In: Proceedings of the 12th International
  Symposium on Algorithmic Game Theory ({SAGT} 2019). Lecture Notes in Computer
  Science, vol. 11801, pp. 63--76. Springer (2019)

\bibitem{LiLYZ20}
Li, M., Lu, P., Yao, Y., Zhang, J.: Strategyproof mechanism for two
  heterogeneous facilities with constant approximation ratio. In: Proceedings
  of the 29th International Joint Conference on Artificial Intelligence
  ({IJCAI} 2020). pp. 238--245 (2020)

\bibitem{LSWZ10}
Lu, P., Sun, X., Wang, Y., Zhu, Z.: {Asymptotically Optimal Strategy-Proof
  Mechanisms for Two-Facility Games}. In: Proc. of the 11th {ACM} Conference on
  Electronic Commerce (EC~'10). pp. 315--324 (2010)

\bibitem{LWZ09}
Lu, P., Wang, Y., Zhou, Y.: Tighter bounds for facility games. In: Proc. of the
  5th Workshop on Internet and Network Economics (WINE~'09). LNCS, vol.~5929,
  pp. 137--148 (2009)

\bibitem{MeiLYZ19}
Mei, L., Li, M., Ye, D., Zhang, G.: Facility location games with distinct
  desires. Discrete Applied Mathematics  \textbf{264},  148--160 (2019)

\bibitem{Meir19}
Meir, R.: Strategyproof facility location for three agents on a circle. In:
  Proceedings of the 12th International Symposium on Algorithmic Game Theory
  ({SAGT} 2019). Lecture Notes in Computer Science, vol. 11801, pp. 18--33.
  Springer (2019)

\bibitem{Miy01}
Miyagawa, E.: Locating libraries on a street. Social Choice and Welfare
  \textbf{18},  527--541 (2001)

\bibitem{NST12}
Nissim, K., Smorodinsky, R., Tennenholtz, M.: Approximately optimal mechanism
  design via differential privacy. In: Innovations in Theoretical Computer
  Science ({ITCS} 2012). pp. 203--213. {ACM} (2012)

\bibitem{PT09}
Procaccia, A., Tennenholtz, M.: Approximate mechanism design without money. In:
  Proc. of the 10th {ACM} Conference on Electronic Commerce (EC~'09). pp.
  177--186 (2009)

\bibitem{Rough19}
Roughgarden, T.: {Beyond Worst-Case Analysis}. Communications of the {ACM}
  \textbf{62}(3),  88--96 (2019)

\bibitem{Rough20}
Roughgarden, T.: {Beyond the Worst-Case Analysis of Algorithms}. Cambridge
  University Press (2020)

\bibitem{RoughOpt}
Roughgarden, T., Talgam{-}Cohen, I.: {Approximately Optimal Mechanism Design}.
  CoRR  (2018), \url{http://arxiv.org/abs/1812.11896}

\bibitem{Rough17_lect6}
Roughgarden, T.: Lecture 6: Perturbation-stable clustering. CS264: Beyond
  Worst-Case Analysis  (2017), \url{http://timroughgarden.org/w17/l/l6.pdf}

\bibitem{SV07}
Schummer, J., Vohra, R.: Mechanism design without money. Algorithmic Game
  Theory  \textbf{10},  243--299 (2007)

\bibitem{SV16}
Serafino, P., Ventre, C.: Heterogeneous facility location without money.
  Theoretical Computer Science  \textbf{636},  27--46 (2016)

\bibitem{SuiBS15}
Sui, X., Boutilier, C., Sandholm, T.: Analysis and optimization of
  multi-dimensional percentile mechanisms. In: Proceedings of the 23rd
  International Joint Conference on Artificial Intelligence ({IJCAI} 2013). pp.
  367--374. {IJCAI/AAAI} (2013)

\end{thebibliography}

\newpage\appendix

\section{The Proof of Lemma~\ref{l:separated-clusters}}
\label{s:app:separated-clusters}

\begin{proof}
It suffices to establish the lemma for two consecutive clusters $C_i$ and $C_{i+1}$. 
We recall that $d(C_i, C_{i+1}) = d(x_{i,r}, x_{i+1,l})$. Moreover, by symmetry, we can assume wlog. that $\Diam(C_i) \geq \Diam(C_{i+1})$. 

If $C_i$ is a singleton, $\Diam(C_i) = 0$ and the lemma holds trivially. If $|C_i| = 2$, wlog. we can only consider the case where $x_{i,l}$ is $C_i$'s center. Otherwise, i.e., if $x_{i,r}$ is $C_i$'s center in optimal clustering $(C_1,\ldots,C_i,\ldots,C_k)$ with centers $(c_i,\ldots,x_{i,r},\ldots,c_j)$, the same clustering $(C_1,\ldots,C_i,\ldots,C_k)$ with centers $(c_1,\ldots,x_{i,l},\ldots,c_j)$ is also optimal for the $\gamma$-stable instance $\vec{x}$ (and should still be optimal after a $\gamma$ perturbation of $\vec{x}$, due to the stability of the instance). We then have:
\begin{equation*}
\begin{split}
 \Diam(C_i) = d(x_{i,l} , x_{i,r}) =d(c_i, x_{i,r}) &< \frac{1}{(\gamma - 1)} d(x_{i,r},x_{i+1,r}) = \frac{1}{(\gamma -1)} d(C_i,C_{i+1}) \Rightarrow \\ 
 d(C_i,C_{i+1})&>(\gamma-1)\Diam(C_i) 
 \end{split}
\end{equation*}

where the first inequality follows from Proposition~\ref{pr:separation}. The lemma then follows by noticing that for any $\gamma\geq 1$:

\[ \gamma-1\geq \frac{\gamma^2+1}{2\gamma } -1 \]

The most interesting case is where $|C_i| \geq 3$ and $x_{i,l} < c_i \leq x_{i,r}$. Suppose $d(x_{i,l}, c_i) = \beta\Diam(C_i)$, for some $\beta \in (0, 1]$ and hence  $d(c_i, x_{i,r}) = (1-\beta)\Diam(C_i)$ (i.e., $\beta$ quantifies how close $c_i$ is to $C_i$'s extreme points and to the closest point of $C_{i+1}$.) We recall that $d(C_i,C_{i+1}) = d(x_{i,r},x_{i+1,l})$.

We start with a tighter analysis of the equivalent of Proposition~\ref{pr:separation} for $x_{i,l}$ and $x_{i+1,l}$, taking into account their specific ordering on the line:
\begin{equation*}
\begin{split}
 d(x_{i,l},x_{i+1,l}) & \geq d(x_{i,l},c_{i+1}) - d(x_{i+1,l}, c_{i+i}) \\
& > \gamma d(x_{i,l},c_i) - \frac{d(x_{i+1,l}, c_i)}{\gamma} \\
& = \gamma d(x_{i,l},c_i) -  \frac{d(x_{i+1,l}, x_{i,l}) - d(x_{i,l},c_i)}{\gamma} \Rightarrow\\
 d(x_{i,l}, x_{i+1,l}) &> \frac{\gamma^2 +1 }{\gamma +1}d(x_{i,l},c_i)
\end{split}
\end{equation*}

Where the second inequality stands due to the \emph{$\gamma$-center proximity property} of $\gamma$ stable instances and the equality stands because $x_{i,l} < c_i < x_{i+1,l}$. Since $d(C_i,C_{i+1}) = d(x_{i,r},x_{i+1,l})= d(x_{i,l},x_{i+1,l})-\Diam(C_i)$, and by $d(x_{i,l},c_i)=\beta\Diam(C_i)$, we get that:

\begin{equation}\label{eq:sepbound_1}
    d(C_i,C_{i+1})>\Big(\frac{\beta(\gamma^2+1)}{\gamma+1} - 1\Big)\Diam(C_i)
\end{equation}

Furthermore, by Proposition\ref{pr:separation}, we have that $d(x_{i,r},x_{i+1,l}) > (\gamma-1)d(x_{i,r},c_i)$. Hence, by $d(c_i, x_{i,r}) = (1-\beta)\Diam(C_i)$, we get that:
\begin{equation}\label{eq:sepbound_2}
    d(C_i,C_{i+1})>(1-\beta)(\gamma-1)\Diam(C_i)
\end{equation}

So, by (\ref{eq:sepbound_1}) and (\ref{eq:sepbound_2}) we have that it must be:
\begin{equation}\label{eq:sepmax}
    d(C_i,C_{i+1})>\max\Big\{\frac{\beta(\gamma^2+1)}{\gamma+1} - 1, (1-\beta)(\gamma-1)\Big\}D(C_i)
\end{equation}

We now observe that for any fixed $\gamma > 1$, the first term of the max in \eqref{eq:sepmax}, $\frac{\beta(\gamma^2+1)}{\gamma+1} - 1$, is increasing for all $\beta > 0$, while the second term, $(1-\beta)(\gamma-1)$, is decreasing for all $\beta \in (0, 1]$. Hence,
for any fixed $\gamma > 1$, the minimum value of the max in \eqref{eq:sepmax} is achieved when $\beta$ satisfies:
\begin{equation*}
    \frac{\beta(\gamma^2+1)}{\gamma+1} - 1 =  (1-\beta)(\gamma-1)
\end{equation*}
Solving for $\beta$, we get that: 
\begin{equation}\label{eq:seprightc}
    \beta= \frac{1}{2}+\frac{1}{2\gamma}\,,
\end{equation}
with $\beta \in (1/2, 1]$, when $\gamma\geq 1$.

We conclude the proof by substituting the value of $\beta$ in \eqref{eq:seprightc} to \eqref{eq:sepmax}. 
%
\qed\end{proof}

\section{The Proof of Lemma~\ref{3_stab_lemma}}
\label{app:s:3_stab_lemma}

\begin{proof}
We establish the lemma for the leftmost agent $x_{i,l}$ as the deviating agent. Specifically, we show that $x_{i,l}$ needs to move by at least $d(x_{i,l},c_i)$ to the left in order to become a singleton cluster. The property then follows for the rest of the agents.

Suppose $x_{i,l}$ can create a singleton cluster by deviating less than $d(x_{i,l},c_i)$ to the left. I.e., for some $x'$ such that $d(x', x_{i,l}) < d(x_{i,l},c_i)$ the optimal clustering of $\vec{x}' = (\vec{x}_{-x_{i,l}}, x')$ is such that the agent location at $x'$ becomes a singleton cluster. We call this clustering (that is optimal for $\vec{x}'$) $\vec{C}'$. Notice that since $d(x', x_{i,l}) < d(x_{i,l},c_i)$, $x'$ is in the gap between clusters $C_{i-1}$ and $C_i$ as by $3$-perturbation stability we have $d(x_{i-1,r},x_{i,l}) > 2d(x_{i,l},c_i)$. This means that in order for this case to be feasible, no agents from $C_{i-1}$ can be clustered together with agents in $C_{i}$ in $(\vec{x}', \vec{C}')$, because $x'$ lies between them and is a singleton cluster.

Consider now the instance $\vec{x}_{-x_{i,l}}$. We know that $cost(\vec{x}_{-x_{i,l}}, \vec{C}')\geq cost(\vec{x}_{-x_{i,l}}, \vec{C})$. That is, since otherwise the optimal clustering for $\vec{x}$ would make $x_{i,l}$ a singleton cluster and serve the rest of the agents as in $\vec{C}'$. Let $\mathrm{diff}$ be the difference in the total cost agents in $\vec{x}_{-x_{i,l}}$ experience between clusterings $\vec{C}$ and $\vec{C}'$. I.e. $\mathrm{diff} = cost(\vec{x}_{-x_{i,l}}, \vec{C}') -  cost(\vec{x}_{-x_{i,l}}, \vec{C})$. As before, since $x_{i,l}$ is not a singleton cluster in $(\vec{x},\vec{C})$ we know that $d(x_{i,l},c_i)<\mathrm{diff}$ (or else setting $x_{i,l}$ as a singleton would have a lower cost in $\vec{x}$ than $\vec{C}$).

But we can perform a $3$-perturbation in $\vec{x}$ in the following way: Scale down all distances between agents from $x_1$ up to $x_{i-1,r}$ and all distances between agents from $x_{i,l+1}$ to $x_n$ ($x_n$ being the rightmost agent of the instance) by $3$. Call this instance $\vec{x}_{per}$. Since agents of clusters $C_{i-1}$ and  $C_{i}$ are not clustered together neither in $\vec{C}$ nor in $\vec{C}'$ we have that
\[\mathrm{diff}_{per} \leq \frac{cost(\vec{x}_{-x_{i,l}}, \vec{C}') -  cost(\vec{x}_{-x_{i,l}}, \vec{C})}{3}\,.\]
So $\mathrm{diff}_{per}\leq \mathrm{diff}/3$. Since $d(x_{i,l}, c_i)$ is unaffected in the perturbation and by stability the optimal clustering of $\vec{x}_{per}$ must remain the same (as $\vec{x}$) we have that it must be $d(x_{i,l},c_i)<\mathrm{diff}/3$ (1).

Finally, the least amount of extra social cost suffered between $cost(\vec{x}, \vec{C})$ and the case of setting $x_{i,l}$ as a center that serves only itself and serve the remaining agents of the instance as on $\vec{C}'$ (i.e. as they would be served should $x'$ gets a facility that served only herself), will be $\mathrm{diff} - d(x_{i,l},c_i)$. 
This means that the optimal clustering algorithm would only choose this solution when $d(x',c_i) > \mathrm{diff} - d(x,c_i)$. So the agent must deviate by at least $\mathrm{diff} - 2d(x,c_i)$. But from (1) we have
\[ \mathrm{diff}-2d(x,c_i) > 3d(x,c_i) -2d(x_{i,l},c_i) = d(x_{i,l},c_i)\,, \]
which concludes the proof of the lemma. 
\qed\end{proof}

\section{The Proof of Lemma~\ref{l:singleton_deviation_5}}
\label{app:s:singleton_deviation_5}

We first present the outline of the proof and then the proof follows. We do this because despite the mostly relatively straight forward arguments used in the proof, due to the delicate formalization required in order to formally describe all the mentioned conditions, the proof gains a good amount of descriptive length. We consider random agent $x_i\in C_i$ of instance $\vec{x}$ with optimal clustering $\vec{C} = (C_1,...,C_k)$, deviating to location $x'$ creating instance $\vec{x}'=(x_{-i},x')$. 

Initially we show that due to the large distance between clusters $C_i$ and $C_j$ with $i\neq j$, guaranteed by Lemma~\ref{l:separated-clusters} for $5$-stable instances, we need only study the cases where $x'\in (x_{i-1,r},x_{i,l})$ and $x'\in (x_{i,r},x_{i+1,l})$ and in the optimal clustering $\vec{C}'$ of instance $\vec{x}'$ no agent in $\vec{x}'\setminus C_i$ is served together with any agent in $C_i$\footnote{Note here that we refer to the group of agents that belong in cluster $C_i$ of the optimal clustering of instance $\vec{x}$. This group is well defined for instance $\vec{x}'$ as well.}, as in all other cases either $x'$ is not a singleton in $\vec{C}'$ or $d(x',x_i)>\Diam(C_i)$. 

The rest of the proof follows the logic of the proof of Theorem~\ref{th:optimal} (which follows), tailored to this specific case. More specifically, given the observation above, we notice the following: In alternative clustering $\vec{C}''$ in which we forcefully place two facilities serving only agents in $C_i$ (optimally with regards to serving agents in $C_i$),  and serve the remaining agents, $\vec{x}\setminus C_i$, optimally with $k-2$ facilities, the cost agents in $\vec{x}\setminus C_i$ experience in clustering $\vec{C}''$ is the same cost agents in $\vec{x}'\setminus C_i\bigcup x'$ experience in clustering $\vec{C}$ (notice that the sets $\vec{x}\setminus C_i$ and $\vec{x}'\setminus C_i\bigcup x'$). Now, the cost of agents in $C_i$ in clustering $\vec{C}''$ is at least $\Diam(C_i)/2$ smaller than it is in $\vec{C}'$ (that is since we can always place the facility to the edge agent further from $c_i$ - see proof of Theorem~\ref{th:optimal}). But since $\vec{C}''$ is not optimal for $\vec{x}$ this means that agents in $\vec{x}\setminus C_i$ experience an increase in cost larger than $\Diam(C_i)/2$ in clustering $\vec{C}''$ when compared to clustering $\vec{C}$. For brevity we symbolize this cost increase as $cst$, so we say $cst>\Diam(C_i)/2$. 

We  now we consider the $4$-perturbation of instance $\vec{x}$ in which all distances among agents to the left and to the right of $C_i$ are shrunk by a factor of $4$. By stability we know that the optimal clustering of the perturbed instance should be the same as the optimal clustering of the original! But in the perturbed instance all costs of agents in $\vec{x}\setminus C_i$ are divided by $4$ in both clusterings $\vec{C}$ and $\vec{C}''$ while the costs of agents in $C_i$ remain the same. So, in order for $\vec{C}''$ to be sub-optimal in the perturbed instance it must be $cst/4>\Diam(C_i)/2$ which means $cst>2\Diam(C_i)$. But serving agent $x'$ of $\vec{x}'$ by $c_i$ has cost at most $2\Diam(C_i)$ if $d(x',x_i)<\Diam(C_i)$ since $d(x_i,c_i)<\Diam(C_i)$. This means that clustering $\vec{C}'$ cannot be optimal for $\vec{x}$.

\begin{proof}
We want to show the lemma for any $\gamma$-stable instance for $\gamma\geq 5$. 

We prove the lemma for random agent $x_j\in C_i$ for some cluster $C_i$ in optimal clustering $
\vec{C}=(C_1,\ldots,C_k)$ of the $\gamma$-stable instance $\vec{x}$. Consider that the agent declares false location $x'$ providing input profile $\vec{x}'=(\vec{x}_{-j}, x')$ to the mechanism in order to become a singleton cluster. That is, if the optimal clustering of instance $\vec{x}'$ is $\vec{Y}$ then $x'$ is a single agent cluster in $\vec{Y}$. 


We first study the case where $|C_i|=2$. But then, from Lemma~\ref{3_stab_lemma} we know that for any $\gamma$-stable instance for $\gamma\geq 3$ agent $x_j\in C_i$ of optimal clustering $\vec{C}$ must deviate by at least his distance to $C_i$'s center in order to become a singleton cluster in $\vec{Y}$. I.e. it must be $d(x',x_j)>d(x_j,c_i)=\Diam(C_i)$, so the lemma stands for this case.

For the most general case, $|C_i|\geq 3$ we start with some observations.
By Lemma~\ref{l:separated-clusters} we know that for any two clusters $C_i$ and $C_j$ of optimal clustering $
\vec{C}=(C_1,\ldots,C_k)$ of $\vec{x}$ we have $d(C_i,C_j)>\Big(\frac{(\gamma-1)^2}{2\gamma} \Big)\max\{\Diam(C_i),\Diam(C_j)\}$. For $\gamma\geq 5$ that is:

\begin{equation}\label{eq:dist_in_5}
d(C_i,C_j)>1.6\max\{\Diam(C_i),\Diam(C_j)\}. 
\end{equation}

Too begin, we notice the following claim:
\begin{myclaim}
\label{cl:obvious}
Agent $x_i$ cannot declare a false location $x'$ with $x_{i,l}\leq x'\leq x_{i,r}$ in such a way that $x'$ is a singleton cluster in $\vec{Y}$.
\end{myclaim}
We can easily see the validity of the claim, since by optimality (also see proof of Theorem~\ref{th:optimal}) $x_j\in C_i$ cannot change the optimal clustering by deviating within the bounds of cluster $C_i$, i.e. if $x_{i,l}\leq x'\leq x_{i,r}$.Hence it must be $x'\neq [x_{i,l},x_{i,r}]$. Even so, for completeness, we provide a proof of the claim, tailored to the case of $5$-stable instances, after the proof of the lemma.

In addition, we notice that if $x'\leq x_{i-1,r}$ or $x'\geq x_{i+1,l}$ then the lemma trivially stands, again by Equation~\ref{eq:dist_in_5} (I.e. in this case it would be $d(x',x_j)>1.6\Diam(C_i)$).

This means that we need only study the cases where $x'\in (x_{i-1,r},x_{i,l})$ or $x'\in (x_{i,r},x_{i+1,l})$ and $d(x',x_j)<\Diam(C_i)$ (and show that $x'$ cannot become a singleton cluster in $\vec{Y}$ in these cases). 

Suppose, contrary to the lemma's claim, that agent declares location $x'\in (x_{i-1,r},x_{i,l})$ with $d(x',x_j)<\Diam(C_i)$ such that $x'$ is a singleton in $\vec{Y}$ (the other case, $x'\in (x_{i,r},x_{i+1,l})$, is symmetrical). Then we notice the following three properties for optimal clustering $\vec{Y}$ of instance $\vec{x}'$:
\begin{description}
    \item[Property 1:] In $\vec{Y}$ there is a facility among agents in ${C_i\setminus x_j}$.
    \item[Property 2:] In $\vec{Y}$ no agent``to the left" of cluster $C_i$ (i.e. by an agent in some cluster $C_l$ for $l<i$, of $\vec{C}$) is served by an agent in $C_i\setminus x_j$ .
    \item[Property 3:] In $\vec{Y}$ no agent  ``to the right" of cluster $C_i$ (i.e. by an agent in some cluster $C_l$ for $l>i$, of $\vec{C}$) is served by an agent in $C_i\setminus x_j$ .
\end{description}

The imminent conclusion from Properties 1, 2 and 3 is the following: Consider instance $\vec{x}\setminus C_i$ and it's optimal $k-2$-clustering $\vec{C}_{-2}$. Then $cost(\vec{x}\setminus C_i, \vec{Y}) = cost(\vec{x}\setminus C_i, \vec{C}_{-2})$\footnote{For a description of this notation, of the form $cost(\vec{x},\vec{C})$, see proof of Theorem~\ref{th:optimal}}. We provide short proofs for each one of these three properties right after the proof of the lemma.

We are now ready to complete the proof. In order to do so we bound the extra cost experienced by agents in $\vec{x}\setminus C_i$ in the possible re-clustering after $x_i$'s deviation, i.e. $cost(\vec{x}\setminus C_i,\vec{Y}) - cost(\vec{x}\setminus C_i,\vec{C})$. We do this by considering the following alternative clustering $C'$ of instance $\vec{x}$: serve agents in $C_i$ using two facilities, optimally and agents in $\vec{x}\setminus C_i$ using the remaining $k-2$ facilities optimally. So in $\vec{C}'$ we have:
\begin{equation}
\label{eq:gain_bound}
    cost(C_i,\vec{C}')\leq cost(C_i,\vec{C}) -\frac{\Diam(C_i)}{2},
\end{equation}
since placing the second facility placed among agents in $C_i$ to the edge-agent further away from $c_i$ reduces the cost by at least $\frac{\Diam(C_i)}{2}$.

But since $\vec{C}$ is optimal in $\vec{x}$ and hence $\vec{C}'$ is not, it must be:
\begin{equation}
\label{eq:loss_bound}
    cost(\vec{x}\setminus C_i,\vec{C'}) - cost(\vec{x}\setminus C_i,\vec{C}) > \frac{\Diam(C_i)}{2}
\end{equation}

Otherwise it would be $cost(\vec{x},\vec{C}')< cost(\vec{x},\vec{C})$. Now notice that properties 1, 2 and 3 mean that agents in $\vec{x}\setminus C_i$ are clustered in exactly the same way in $\vec{C'}$ as in $\vec{Y}$. That means that:
\begin{equation}
\label{eq:equality_of_clusterings}
    cost(\vec{x}'\setminus \{C_i\bigcup x'\}, \vec{Y}) = cost(\vec{x}\setminus C_i, \vec{C}')
\end{equation}
and that no agent to the left of $C_i$ is clustered together with any agent to the right of $C_i$ in $\vec{C}'$.

The last observation means if we consider a $4$-perturbation of instance $\vec{x}$, instance $\vec{x}_p$, in which we divide all distances among agents between $[x_l,x_{i-1,r}]$ and agents between $[x_{i+1,l},x_r]$, where $x_l$ and $x_r$ the leftmost and rightmost agents of the instance equivalently we have that:
\begin{equation*}
    cost(\vec{x}_p\setminus C_i, \vec{C'}) - cost(\vec{x}_p\setminus C_i, \vec{C}) = \frac{cost(\vec{x}\setminus C_i,\vec{C'}) - cost(\vec{x}\setminus C_i,\vec{C})}{4}
\end{equation*}

But in $x_p$ the distances among agents in $C_i$ remain unaffected which means that in $x_p$, Equation~\ref{eq:gain_bound} still stands. This means, that since the instance is $5$-stable, clustering $\vec{C'}$ must still be sub-optimal in $\vec{x}_p$ and hence it must be 

\begin{equation}
\label{eq:final_bit}
\begin{split}
    cost(\vec{x}_p\setminus C_i,\vec{C'}) - cost(\vec{x}_p\setminus C_i,\vec{C}) &> \frac{\Diam(C_i)}{2}\Rightarrow \\
    \frac{cost(\vec{x}\setminus C_i,\vec{C'}) - cost(\vec{x}\setminus C_i,\vec{C})}{4} &> \frac{\Diam(C_i)}{2}\Rightarrow \\
    cost(\vec{x}\setminus C_i,\vec{C'}) - cost(\vec{x}\setminus C_i,\vec{C})&> 2\Diam(C_i).
    \end{split}
\end{equation}

Noticing again that by Equation~(\ref{eq:equality_of_clusterings}), $    cost(\vec{x}'\setminus \{C_i\bigcup x'\}, \vec{Y}) = cost(\vec{x}\setminus C_i, \vec{C}')$ and by Equation~(\ref{eq:final_bit}) and $cost(\vec{x}\setminus C_i,\vec{C}) = cost(\vec{x}'\setminus \{C_i\bigcup x'\},\vec{C})$ we have

\begin{equation*}
    cost(\vec{x}'\setminus \{C_i\bigcup x'\},\vec{Y}) - cost(\vec{x}\setminus \{C_i\bigcup x'\},\vec{C})> 2\Diam(C_i).
\end{equation*}
Finally, since $d(x',x_i)<\Diam(C_i)$, 
\begin{equation*}
    cost(\{C_i\bigcup x' \setminus x_i\}, \vec{C})-cost(\{C_i\bigcup x' \setminus x_i\}, \vec{Y}) < 2\Diam(C_i),
\end{equation*}
since $d(x_i,c_i)\leq\Diam(C_i)$. By adding the last two equations we get that $cost(\vec{x}',\vec{Y})> cost(\vec{x}',\vec{C})$ which means that $\vec{Y}$ is not optimal.

\qed
\end{proof}

We now present the proofs of \textbf{Claim~\ref{cl:obvious}} and \textbf{Properties 1, 2 and 3}, used in the main proof of Lemma~\ref{l:singleton_deviation_5}. 

\begin{proof}[Of Claim~\ref{cl:obvious}]
 Consider $x_{i,l}'$ and $x_{i,r}'$ to be the leftmost and rightmost agents of $C_i\setminus x_j$ (i.e. if $x_j\neq x_{i,r},x_{i,l}$ then $x_{i,r}=x_{i,r}'$ and $x_{i,l}=x_{i,l}'$).

Contrary to the claim, suppose $x_{i,l}'\leq x'\leq x_{i,r}'$ and $x'$ is a singleton cluster  $\vec{Y}$. Since $x'$ is a singleton and $x_{i,l}'$ and  $x_{i,r}'$ are to her left and right side equivalently, $x_{i,l}'$ and  $x_{i,r}'$ cannot be served by the same facility in $\vec{Y}$ (since clustering $\vec{Y}$ is optimal for $\vec{x}'$). This means that either $x_{i,l}'$ or  $x_{i,r}'$ is served by an agent in $\vec{x}\setminus C_i$ or there are two facilities among agents in $C_i\setminus x_j$ in $\vec{Y}$. Both of these cases are infeasible though. For the first one, suppose that $x_{i,r}$ is not served by an agent in $C_i\setminus x_j$. By Equation~\ref{eq:dist_in_5} that means that the cost of serving $x_{i,r}$ is at-least $1.6\Diam(C_i)$. But since $x_{i,l}'\leq x'\leq x_{i,r}'$ $x'$, $d(x',x_{i,r})<\Diam(C_i)$ so $\vec{Y}$ could not be optimal in $\vec{x}'$. For the latter case ($\vec{Y}$ places two facilities among agents in $C_i\setminus x_j$) we see that if $\vec{Y}$ is optimal for $\vec{x}'$ then the optimal $(k-1)$-clustering of instance $(\vec{x}_{-j})$ would place two facilities among agents in $C_i\setminus x_j$ (since $x'$ is a singleton removing her and one facility from the instance should yield the exact same clustering for the rest of the agents). But then, since in $\vec{C}$ there is only one facility among agents in $C_i$, $\vec{C}$ could not be optimal for instance $\vec{x}$ (because if the optimal $(k-1)$-clustering of instance $(\vec{x}_{-j})$ places two facilities among agents in set $C_i\setminus x_j$ then the optimal $k$-clustering of instance $\vec{x}$ should place at least as many among agents in $C_i$) , which is a contradiction.

Finally we notice that if $x_j=x_{i,l}$, $x'$ cannot become a singleton in $\vec{Y}$ if $x'\in [x_{i,l}, x_{i,l}']$ since the cost serving agent $x_j$ by $c_i$ in that interval is only decreased (in relation to the cost of serving her by $c_i$ in $\vec{x}$ - she's getting closer to her serving facility). Similarly for the case of $x_j=x_{i,r}$ moving in interval $[x_{i,r}',x_{i,r}]$. The above mean that agent $x_j$ cannot become a singleton cluster by moving within the bounds of $C_i$ (i.e. if $x'$ is a singleton in $\vec{Y}$  it must be $x'\notin [x_{i,l},x_{i,r}]$), which is the claim.
\qed
\end{proof}

\begin{proof}[Of Property 1]
We know that $|C_i\setminus x_j|\geq 2$. Furthermore, since $d(x',x_{j})<\Diam(C_i)$ we have that $d(x',c_i)<2\Diam(C_i)$. But if there is no facility among agents in $C_i\setminus x_j$ that means that these agents are all served by a facility placed in a location $x_l$ with $x_l\in C_l$ with $l\neq i$. But, again by Equation~\ref{eq:dist_in_5} that would mean that 
\begin{equation}\label{eq:fac_in_clust_1}
cost(C_i\setminus x_j, \vec{Y})>2*1.6\Diam(C_i) + cost(C_i\setminus x_j, \vec{C})
\end{equation}
(since $|C_i\setminus x_j|\geq 2$, $d(C_i,C_l)\geq 1.6\Diam(C_i)$). Furthermore, since agents in $\vec{x}\setminus C_i$ are served by the same number of facilities in $\vec{Y}$ as in $\vec{C}$, but also have to serve agents in ${C_i\setminus x_j}$ in $\vec{Y}$ (i.e. the placement of the $(k-1)$ facilities among agents in $\vec{x}\setminus C_i$ is not optimal in $\vec{Y}$ as it is in $\vec{C}$, for these agents) we have
\begin{equation}\label{eq:fac_in_clust_2}
cost(\vec{x}\setminus C_i, \vec{Y}) \geq cost(\vec{x}\setminus C_i, \vec{C}).
\end{equation}

Hence, by adding (\ref{eq:fac_in_clust_1}) and (\ref{eq:fac_in_clust_2}) we have that :

\begin{equation*}
    cost(\vec{x}\setminus x_j,\vec{Y})> 2*1.6\Diam(C_i)+cost(\vec{x}\setminus x_j,\vec{C})
\end{equation*}

By remembering that $\vec{x}' = (\vec{x}_{-j},x')$ and in $\vec{Y}$ $x'$ is a singleton cluster (i.e. has cost 0) the above becomes:

\begin{equation}\label{eq:singleton_clustering_Y}
    cost(\vec{x}',\vec{Y})> 2*1.6\Diam(C_i)+cost(\vec{x}\setminus x_j,\vec{C})
\end{equation}

But, alternative clustering $\vec{C}'$ for $\vec{x}'$ in which we serve all agents as we do in $\vec{C}$ and also serve location $x'$ by $c_i$ has cost 
\begin{equation}\label{eq:singleton_clustering_C}
cost(\vec{x}', \vec{C}') \leq cost(\vec{x}\setminus x_j,\vec{C}) + 2\Diam(C_i),
\end{equation}
 since $d(x',c_i)<2\Diam(C_i)$.

This means that, by (\ref{eq:singleton_clustering_Y}) and (\ref{eq:singleton_clustering_C}) $cost(\vec{x}',\vec{Y})>cost(\vec{x}',\vec{C'})$ which means that clustering $\vec{Y}$ would be sup-optimal for instance $\vec{x}'$, which is a contradiction.

Notice that by Observation~\ref{obs:cluster_distances}, property 1 means that no agents in $C_i\setminus x_j$ are served by an agent not in $C_i$ in  $\vec{Y}$.
\qed
\end{proof}

\begin{proof}[Of Property 2]
Property 2 is trivial: since $x'\in (x_{i-1,r},x_{i,l})$ and $x'$ forms a singleton cluster in $\vec{Y}$, by optimality no agent to the left of $x'$ is clustered together with agents to the right of $x'$.
\qed
\end{proof}

\begin{proof}[Of Property 3]
Initially, for property 3 we notice the following: At most 1 agent in $C_{i+1}$ can be clustered together with agents in $C_i$ in $\vec{Y}$. Otherwise, due to the distance between $C_i$ and $C_{i+1}$, clustering $\vec{Y}$ would be sub-optimal (using the same reasoning as for property 1). Obviously, due to optimality, this agent can only be $x_{i+1,l}$. 

We now consider the structure of cluster $C_i$ in relation to agent $x_{i+1,l}$. Specifically, by Equation~(\ref{eq:dist_in_5}) it must be 
\begin{equation}
\label{eq:min_dist_of_right_agent}
d(x_{i,r},x_{i+1,l})>1.6\cdot \Diam(C_i),
\end{equation}
 since $d(x_{i,r},x_{i+1,l})=d(C_i,C_{i+1})$. By looking at the proof of Lemma~\ref{l:separated-clusters} we see that the smallest possible distance between $C_i$ and $C_{i+1}$ is achieved when $d(x_{i,l},c_i)=\frac{\Diam(C_i)}{c}$ for $c=\frac{2\gamma^2}{\gamma^2+\gamma} \Rightarrow \frac{1}{c}=0.6$ for $\gamma=5$. This means that since agent $x_j$ deviates to the left in this case, by at most $\Diam(C_i)$, it must be 
\begin{equation}
\label{eq:furthest_deviation_possible}
d(x',c_i)\leq 1.6\Diam(C_i),
\end{equation}
 in the edge case. 
Furthermore, by Observation~\ref{obs:cluster_distances}, since $x_{i+1,l}$ is not served by an agent in $C_{i+1}$ there is no facility among agents in $C_{i+1}$ in $\vec{Y}$. I.e. all agents in $C_{i+1}\setminus x_{i+1,l}$ are served by a facility placed on $[x_{i+2,l},x_n]$ where $x_n$ the rightmost agent location in the instance. But, by Lemma~\ref{l:separated-clusters}, if $x_{i+1,l}$ is served by $c_{i+1}\in C_{i+1}$ in $\vec{C}$, $d(C_{i+1},C_{i+2})>1.6\Diam(C_{i+1})\geq 1.6 d(x_{i+1,l},c_{i+1})$ and so, it is 
\begin{equation}
\label{eq:min_distance_of_clusters_to_the_right}
cost(x_{i+1,o}, \vec{Y})\geq d(x_{i+1,o},x_{i+2,l})\geq 1.6 d(x_{i+1,l},c_{i+1}),
\end{equation}
 for every $x_{i+1,o}\in C_{i+1}\setminus x_{i+1,l}$.

Now we are able to show that clustering $\vec{Y}$ cannot be optimal for instance $\vec{x}'$ in the edge case. We will compare it with clustering $\vec{C}$ (where every agent is served by the same facility as in clustering $\vec{C}$ and $x'$ is served by $c_i$). We have the following:

\begin{equation*}
    cost(\vec{x}'\setminus\{x'\bigcup C_{i}\bigcup C_{i+1}\}, \vec{Y})\geq cost(\vec{x}'\setminus\{x'\bigcup C_{i}\bigcup C_{i+1}\}, \vec{C}),
\end{equation*}
 by optimality. Furthermore,
\begin{equation*}
\begin{split}
    cost(C_{i+1}, \vec{Y}) &\geq cost(C_{i+1}, \vec{C}) - d(x_{i+1,l},c_i) + 1.6 d(x_{i+1,l},c_{i+1})+1.6\Diam(C_i),
\end{split}
\end{equation*}
 by optimality and equations (\ref{eq:min_dist_of_right_agent}) and (\ref{eq:min_distance_of_clusters_to_the_right}). Also, 
\begin{equation*}
    cost(C_{i}\setminus x, \vec{Y})\geq cost(C_{i}\setminus x, \vec{C}),
\end{equation*}
 by optimality. Finally,
\begin{equation*}
cost(x',\vec{Y})+1.6\Diam(C_i) > cost(x',\vec{C}),
\end{equation*}
by equation (\ref{eq:furthest_deviation_possible}).

By adding we get $cost(\vec{x}',\vec{Y})>cost(\vec{x}',\vec{C})$ which means that $\vec{Y}$ is sub-optimal for instance $\vec{x}'$. All we need to finalize this observation is realize that as we move away from the edge case, the above inequalities become easier to satisfy. Specifically if $C_i$ had center $c_i'<c_i$ we see that factor 1.6 of inequality (\ref{eq:furthest_deviation_possible}) decreases while $d(C_i,C_{i+1})$ increases. If $c_i'>c_i$ the same factor of inequality (\ref{eq:furthest_deviation_possible}) may increase by $|c_i'-c_i|$, but then $d(C_i,C_{i+1})$ increases by at least $\frac{\gamma^2+1}{\gamma+1}\cdot |c_i'-c_i| > 4.3|c_i'-c_i|$ (since $d(x_{i,l},x_{i+1,l})>\frac{\gamma^2+1}{\gamma+1}d(x_{i,l},c_i)$ - see proof of Lemma~\ref{l:separated-clusters}), hence maintaining $cost(\vec{x}',\vec{Y})>cost(\vec{x}',\vec{C})$.
\qed
\end{proof}

\section{Proofs of Auxiliary Lemmas Used in the Proof of Theorem~\ref{thm:lower_bound}}
\label{s:app:well-separated}

For completeness, we restate the proofs of the auxiliary lemmas with the properties of well-separated instances adapted from \cite{FT12} and used in the proof of Theorem~\ref{thm:lower_bound}. 

Before we proceed with the proofs of the auxiliary lemmas, we need the following basic fact about the facility allocation of any determistic strategyproof mechanism. 

\begin{lemma}[Proposition 2.1, \cite{FT12}]
\label{l:basicAllocations}
Let $M$ be a deterministic strategyproof with a bounded approximation ratio of $\rho \geq 1$ for $\sqrt{2}$-stable instances of $k$-Facility location on the line. For any $(k+1)$-location instance $\vec{x}$ with $x_1\leq x_2\leq \ldots \leq x_{k+1}$, $M_1(\vec{x})\leq x_{2}$ and $M_k(\vec{x})\geq x_k$.
\end{lemma}

\begin{proof}
We show it for $M_1(\vec{x})\leq x_{2}$, the other case is symmetric. Suppose $x_2<M_1(\vec{x})$. Then the agent in $x_1$ has the incentive to deviate to location $x_2$, since $M_1(\vec{x}_{-1}, x_2) = x_2$ due to the bounded approximation of $M$ (i.e., in $(\vec{x}_{-1}, x_2)$, $M$ allocates $k$ facilities to $k$ different locations). Notice that $(\vec{x}_{-1}, x_2)$ is $\gamma$-stable for any $\gamma \geq 1$.
\end{proof}

\subsection{The Proof of Lemma~\ref{l:nice-allocation}}

\begin{proof}
Since $M$ has a bounded approximation, the isolated pair $x_k$ and $x_{k+1}$ must be served by the same facility $M_k(\vec{x})$. By Lemma~\ref{l:basicAllocations}, we know that $M_k(\vec{x})\geq x_k$. Then, it must also be  $M_k(\vec{x})\leq x_{k+1}$\,. Otherwise, like in Lemma~\ref{l:basicAllocations}, agent $x_k$ could declare location $x_{k+1}$ and decrease her cost, since $M_k(\vec{x}_{-k}, x_{k+1}) = x_{k+1}$ by the bounded approximation of $M$. Again, the instance $(\vec{x}_{-k}, x_{k+1})$ is arbitrarily stable.
\end{proof}

\subsection{The Proof of Lemma~\ref{l:movingRight}}

We can now proceed to the proofs of the auxiliary lemmas, Lemma~\ref{l:movingLeft} and Lemma~\ref{l:movingRight}, which refer to the movement of isolated pairs. We only present the proof of Lemma~\ref{l:movingRight} here. The proof of Lemma~\ref{l:movingLeft} is fully symmetric. 

The proof shown here, refers to $2$-Facility Location on well separated instances with $3$ agents. All arguments as well as the stability factor of the instance only depend on the well separated property of the rightmost pairs of agents as well as their distance from the third agent from the right. That is, that since in all instances studied in the proof we only change distance between the agents of the isolated, rightmost pair, in the range $(0,d(x_1,x_2)/r)$ and only increase the distance between the isolated pair and the leftmost agent $x_1$, any instance with a large enough distance between $x_1$ and $x_2$, i.e. for which $d(x_1,x_2)>\gamma\cdot\rho d(x_2,x_3)$ will be $\gamma$-stable in all parts of the proof.  In that way it is easy to verify that the arguments presented here extend to $(x_1|\ldots|x_{k-1}|x_k,x_{k+1})$-well separated and stable instances of at least a specific minimum distance $d(x_{k-1},x_k)$. 

Consider $M$ to be a deterministic, strategyproof, anonymous and bounded approximation mechanism, with approximation ration of at most $\rho$, for 2-facility location. We will work on instance $\vec{x}$ with three agents $x_1<x_2<x_3$ which is $(x_1|x_2,x_3)$-well separated.

The proof of Lemma~\ref{l:movingRight} directly follows from the following  propositions, originally established in \cite[Appendix~A]{FT12}.

\begin{proposition}
\label{p:moveIsolatedLeft}
Consider $(x_1|x_2,x_3)$-well separated, stable instance $\vec{x}$ for which $M_2(\vec{x}) = x_2$. Then for instance $\vec{x}'=(\vec{x}_-2,x_2')$ where $x_2\leq x_2' \leq x_3$ it will be $M_2(\vec{x'})= x_2'$ 
\end{proposition}

\begin{proof}
Notice that since $d(x_2',x_3)<d(x_2,x_3)$ instance $\vec{x'}$ is still $(x_1|x_2,x_3)$-well separated. Furthermore, since $x_1$ is allocated a facility (by the $\rho$-approximation property of the instance), $\vec{x}'$ is at least as stable as $\vec{x}$ since the distance between the isolated pair is shortened and their distance from $x_1$ has grown. All that needs to be shown is that image set $I_2(\vec{x}_{-2})$ includes the interval $[x_2,x_3]$. Since $x_2$ is allocated a facility, we know $x_2\in I_2(\vec{x}_{-2})$. Furthermore, by the bounded approximation property of $M$ $x_3\in I_2(\vec{x}_{-2})$.
Assume there is a hole $(l,r)\in I_2(\vec{x}_{-2})$ with $x_2\leq l<r\leq x_k$. Consider location $y\in (l,r)$ such that $d(y,l)<d(y,r)$. By strategyproofness $l\in M(\vec{x}_{-2},y)$. But then, by Lemma~\ref{l:basicAllocations} we have that $F_2(\vec{x}_{-j},y)>y$ which contradicts $M$'s bounded approximation ratio, since the two agents of the isolated pair of $(\vec{x}_{-j},y)$ are served by different facilities.
\end{proof}

\begin{proposition}
\label{p:moveIsolatedRightmost}
Consider $(x_1|x_2,x_3)$-well separated stable instance $\vec{x}$ for which $M_2(\vec{x}) = x_2$. Then for every $(x_1|x_2,x_3')$-well separated instance $\vec{x'} = \vec(x_{-3},x_3')$, if $\vec{x'}$ is also well separated, $M_2(\vec{x'})=x_2$.
\end{proposition}

We notice that in that case, the distance between the agents of the isolated pair might grow a from $\vec{x}$ to $\vec{x'}$. Since the proof of this proposition uses instances where the distance of the isolated pair varies from $\epsilon$ to $d(x_1,x_2)/\rho$ the proposition stands for stable instances only if all possible $(x_1|x_2,x_3')$-well separated instances $\vec{x'} = \vec(x_{-3},x_3')$ are well separated. It is easy to see, that since in all these instances it must be $d(x_2,x_3)<d(x_1,x_2)/\rho$ then for a large enough distance $d(x_1,x_2)$ (i.e. $d(x_1,x_2)>\gamma\cdot\rho d(x_2,x_3)$) $\vec{x'}$ is always stable. We show the following proof considering that we have made this assumption.
\begin{proof}
Since $M_2(\vec{x})<x_3$, we know that $x_3\notin I_3(\vec{x_{-3}})$. So, there is a $x_3$-hole $(l,r)\in I_3(\vec{x_{-3}})$. Since $M_2(\vec{x})=x_2$, $l=x_2$ and $r>2x_3-x_2$ (by strategyproofness). By strategyproofnes, if $x_3'<(r+l)/2$ (for $x_2<x_3'$ for well separated instance $\vec{x'}$), $M_2(\vec{x})=x_2$.

To finish, we show that there are no $(x_1|x_2,x_3')$-well separated instances $\vec{x'} = (\vec{x_{-3}},x_3')$ with $x_3'\geq (r+l)/2$ and $M_2(\vec{x'})\neq x_2$. Again, we reach a contradiction by assuming that there is a point $y\geq (r+l)/2$ for which $(\vec{x_{-3}},y)$ is a $(x_1|x_2,y)$-well separated instance with $M_2((\vec{x_{-3}},y))\neq x_2$. If such a $y$ exists, then there exists $x_k'\in [(r+l)/2,r)$ for which $\vec{x'} = (\vec{x_{-3}}, x_3')$ is a $(x_1|x_2,x_3')$-well separated. But then, $M_2(\vec{x'}) = r >x_3'$ (by strategyproofness, because $x_3'$ is closer to $r$ than to $l$). Since $\vec{x'}$ is $(x_1|x_2,x_3')$-well separated this contradicts lemma \ref{l:nice-allocation} which dictates that it must be $M_2(\vec{x_{-3}},x_3')\in [x_2,x_3']$.
\end{proof}

\begin{proposition}
\label{p:moveIsolatedBothRight}

Consider $(x_1|x_2,x_3)$-well separated stable instance $\vec{x}$ for which $M_2(\vec{x}) = x_2$. Then for every $(x_1|x_2',x_3')$-well separated instance $\vec{x'} = (\vec{x}_{-\{2,3\}},x_2',x_3')$, with $x_2<x_2'<(x_2+x_3)/2$, if $\vec{x'}$ is also well separated, $M_2(\vec{x'})=x_2$. 
\end{proposition}

Note that, as for proposition \ref{p:moveIsolatedRightmost} the restriction that $\vec{x'}$ is also $\gamma$-stable  is equivalent to $d(x_1,x_2)>\gamma\cdot\rho d(x_2,x_3)$.

\begin{proof}
Since $x_2'\in [x_2,x_3]$ we have that $M_2(\vec{x_{-2}},x_2') = x_2'$, by proposition \ref{p:moveIsolatedLeft}. But since $d(x_2',x_3)<d(x_2,x_3)$, $(\vec{x_{-2}},x_2')$ is $(x_1|x_2',x_3)$-well separated. Hence, by proposition \ref{p:moveIsolatedRightmost}, for $(x_1|x_2',x_3')$-well separated instance $\vec{x'} = (\vec{x}_{-\{2,3\}},x_2',x_3')$, $M_2(\vec{x'})=x_2'$
\end{proof}

\begin{proposition}
Consider $(x_1|x_2,x_3)$-well separated stable instance $\vec{x}$ for which $M_2(\vec{x}) = x_2$. Then for every $(x_1|x_2',x_3')$-well separated instance $\vec{x'} = (\vec{x}_{-\{2,3\}},x_2',x_3')$, with $x_2\leq x_2'$, if $\vec{x'}$ is also well separated, $M_2(\vec{x'})=x_2$. 
\end{proposition}

\begin{proof}
We will inductively use proposition \ref{p:moveIsolatedBothRight} to create instance $\vec{x}'$. Consider $d=d(x_2',x_2)$, $\delta = d(x_3,x_2)/2$ and $\kappa=\lceil{d/\delta}\rceil$. Then for every $\lambda = 1,2,3\ldots,\kappa$ consider instance $\vec{x}_\lambda = (\vec{x}_{-\{2,3\}},x_2+(\lambda-1)\delta,x_3+(\lambda-1)\delta)$. Now observe that $\vec{x}_\lambda$ is well separated since for it's rightmost pair, $x_2'=x_2+(\lambda-1)\delta$ and $x_3'=x_3+(\lambda-1)\delta$ it is $d(x_2',x_3')>2\delta$ while $d(x_1,x_2')>d(x_1,x_2)$. By iteratively applying proposition \ref{p:moveIsolatedBothRight} to $\vec{x_\lambda}$, we have that for every $(\vec{x}_{-\{2,3\}},y_2,y_3)$ well separated instance with $x_2+(\lambda-1)\delta\leq y_2\leq x_2+\lambda\delta$, $M_2(\vec{x}_{-\{2,3\}},y_2,y_3) = y_2$. For $\lambda=\kappa$ we get $M_2(\vec{x}_{-\{2,3\}},x_2',x_3') = x_2'$\,.
\end{proof}

\section{Missing Details from the Proof of Theorem~\ref{thm:lower_bound}: Case 2}
\label{s:app:lower_bound}

Next, we present a detailed proof of Case~2 in the proof of  Theorem~\ref{thm:lower_bound}.

\noindent\textbf{Case 2: $x_4 - l \leq \sqrt{2}\lambda$.}
Let $m = (r+l)/2$ be the midpoint of the $x_3$-hole $(l, r)$ in $I_3(\vec{x}_{-3})$. We consider the instance $\vec{y} = (\vec{x}_{-3}, a)$, where $a < m$ is arbitrarily close to $m$ (i.e., $a \lesssim m$) so that $a - l < r - a$ and $d(a, x_4) \lesssim \sqrt{2}\lambda/2$. The latter is possible since $x_3$ is already arbitrarily close to $x_4$ and the right endpoint $r$ of the hole $h = (l, r)$ lies in $(x_3, x_4]$. Since $d(x_1, x_2) = \lambda$, $d(x_2, a)$ is quite large, and $d(a, x_4) \lesssim \sqrt{2}\lambda/2$, the instance $\vec{y}$ is $(\sqrt{2}-\delta)$-stable, for any $\delta > 0$. By strategyproofness, $M(\vec{y})$ must place a facility at $l$, since $l \in I_3(\vec{x}_{-3})$ and $l$ is the nearest endpoint of the hole $h = (l, r)$ to $a$. 

As before, we now consider the instance $\vec{y}' = (\vec{y}_{-4}, l)$. Since $d(x_1, x_2) = \lambda$, $d(x_2, a)$ is quite large, and $d(a, l) < d(a, r) \leq \sqrt{2}\lambda/2$, the instance $\vec{y}'$ is $(\sqrt{2}-\delta)$-stable, for any $\delta > 0$. Hence, by strategyproofness, $M(\vec{y}')$ must keep a facility at $l$, because $l \in I_4(\vec{y}_{-4})$. 

To conclude the proof, we need to construct a $(x_1 | x_2 | l', l'+\eps)$-well-separated instance $\vec{z}$ with $l' \in M(\vec{z})$. Then, we can reach a contradiction to the hypothesis that $M$ has a bounded approximation ratio, by applying Lemma~\ref{l:movingLeft}, similarly to Case~1. 

To this end, we consider the image set $I_4(\vec{y}'_{-4})$ of agent $4$ in $\vec{y}'_{-4} = (x_1, x_2, a)$. Since $l \in M(\vec{y}')$, $l \in I_4(\vec{y}'_{-4})$. 
If $a-\eps \in I_4(\vec{y}'_{-4})$, the instance $\vec{z} = (\vec{y}'_{-4}, a-\eps)$ is $(x_1|x_2|a-\eps,a)$-well-separated (and thus, $(\sqrt{2}-\delta)$-stable, for any $\delta > 0$). Moreover, by strategyproofness, $M(\vec{z})$ must place a facility at $a-\eps$, because $a-\eps \in I_4(\vec{y}'_{-4})$. 
Otherwise, there must be a hole $h' = (l', r')$ in the image set $I_4(\vec{y}'_{-4})$, with $l' > l$ (because $l \in I_4(\vec{y}'_{-4})$) and $r' < a-\eps$ (because of the hypothesis that $a -\eps \not\in l \in I_4(\vec{y}'_{-4})$). We consider the instance $\vec{z}' = (\vec{y}'_{-4}, l'+\eps) = (x_1, x_2, l'+\eps, a)$. Since $l' + \eps \in (l, a)$, $d(a, l'+\eps) < d(a, l) < \sqrt{2}\lambda/2$ and the instance $\vec{z}'$ is $(\sqrt{2}-\delta)$-stable, for any $\delta > 0$. Therefore, by strategyproofness and since $l' \in I_4(\vec{y}'_{-4})$, $M(\vec{z}')$ must place a facility at $l'$. We now consider the instance $\vec{z} = (\vec{z}'_{-3}, l') = (x_1, x_2, l', l'+\eps)$, which is $(x_1 | x_2 | l', l'+\eps)$-well-separated (and thus, $(\sqrt{2}-\delta)$-stable, for any $\delta > 0$). Moreover, by strategyproofness and since $l' \in M(\vec{z}')$, and thus, $l' \in I_3(\vec{z}'_{-3})$, $M(\vec{z})$ must place a facility at $l'$. 

Therefore, starting from the $(\sqrt{2}-\delta)$-stable instance $\vec{y}'$, with $l \in M(\vec{y}')$, we can construct a $(x_1 | x_2 | l', l'+\eps)$-well-separated instance $\vec{z}$ with $l' \in M(\vec{z})$. 
Then, by Lemma~\ref{l:movingLeft}, $z_4 = l'+\eps \in M(\vec{z})$, because for  the $(x_1 | x_2 | x_3, x_4)$-well-separated instance $\vec{x}$,  $M_3(\vec{x}) = x_4$, and $\vec{z}$ is a $(x_1 | x_2 | l', l'+\eps)$-well-separated instance with $z_4 \leq x_4$. Since both $l', l'+\eps \in M(\vec{z})$, the social cost of $M(\vec{z})$ is arbitrarily larger than $\eps$, which is the optimal social cost of the $3$-Facility Location instance $\vec{z}$.\qed

\section{The Proof of Theorem~\ref{thm:random}}
\label{app:thm:random}

\begin{proof} 
The approximation guarantee easily follows from the fact that since a facility is uniformly at random placed over each optimal cluster, the expected cost of the sum of the cost of the agents in each cluster is 2 times their cost in the optimal clustering.

As is it always with our mechanisms, agent $x_i\in C_i$ cannot gain by moving within the range of $C_i$ (this would only increase her utility).

Since the analysis of \rand\ is so similar to the analysis of the mechanism in Section~\ref{s:rightmost}, we skip the detailed case analysis and mention only the key deviation cases that need be covered. Specifically these include:

\begin{description}
\item[Case 1:] why agent $x_i\in C_i$ cannot gain by becoming a member of another cluster, 
\item[Case 2:] or by becoming a self serving center 

\item[Case 3:] or by merging or splitting $C_i$.
\end{description}

Without loss of generality, consider the deviating agent to be the edge agent $x_{i,l}\in C_i$, declaring location $x'$ creating instance $\vec{x}'=(\vec{x}_{-x_{i,l}},x')$ with optimal clustering $\vec{C}'$. If our results stand for her, they easily transfer to all agents in $C_i$. $C_i$ contains $n$ agents, including $x_{i,l}$. For simplicity, without loss of generality we index these agents from left to right, excluding $x_{i,l}$\,, such as $x_{i,l}\leq x_{i,1}\leq \cdots \leq x_{i,n-1}$\,, where $x_{i,1}=x_{i,l+1}$ and $x_{i,n-1} = x_{i,r}$\,. Now for simplicity, we represent $d(x_{i,l}, x_{i,j})$ by $d_{i,j}$. Of course $d_{i,l}=0$. We define as $X_i$ the discrete random variable that takes values from sample space $\{d_{i,l},d_{i,1}, d_{i,2}, \ldots, d_{i,n-1}\}$ uniformly at random. That is, $X_i$ represents the cost agent $x_{i,l}$ experiences if she is served by the facility placed in $C_i$ by the mechanism. Then, the expected cost of $x_{i,l}$ should she not deviate is:

\begin{eqnarray*}
\Exp(X_i) = \frac{0+d_{i,1}+ \ldots + d_{i,n-1}}{n}
\end{eqnarray*}

That is, since for any agent $x_j \notin C_i$, $d(x_j,x_{i,l})>D(C_i)=d_{i,n-1}$ by Lemma~\ref{l:separated-clusters}.

Now, for Case~1, ``why agent $x\in C_i$ cannot gain by becoming a member of another cluster''. Notice that this is the case where agents in $C_i$ are not merged or splitted in $\vec{C}'$. With some abuse of notation, this allows us to refer to the cluster containing agents in $C_i\setminus x_{i,l}$ in $\vec{C}'$ of $\vec{x}$ as $C_i'$. $C_{i-1}'$ then is the set of agents belonging to the cluster immediately to the left of $C_i'$ (i.e. the rightmost agent of $C_{i-1}'$, excluding $x'$, is $x_{i-1,r}$). Consider a deviation $x'$ that places the deviating agent in cluster $C_{i-1}'$ after step 1 of the mechanism. Again for simplicity consider $d(x', x_{i,l}) = c$ and we index agents in $C_{i-1}$ inversely, such that $x_{i-1,\hat{1}}\geq x_{i-1,\hat{2}}\geq \ldots \geq x_{i-1,\hat{n'}}$  (meaning that now $x_{i-1,r}=x_{i-1,\hat{1}}$, $x_{i-1,r-1} = x_{i-1,\hat{2}}$ etc.) where $|C_{i-1}| = n'$. Equivalently we set $d(x_{i,l}, x_{i-1,\hat{j}}) = d_{i-1,j}$. By Corollary~\ref{c:separated-clusters}, we have $d_{i,1}\leq d_{i,2}\leq \cdots\leq d_{i,n-1}\leq d_{i-1,1}\leq \cdots \leq d_{i-1,n'}$. Now we define uniform random variable $X_i'$ with sample space $\{d_{i,1},\ldots,d_{i,n-1}\}$ (see that $d_{i,l}$ is now absent) and random variable $X_{i-1}'$ with sample space $\{c, d_{i-1,1},\ldots,d_{i-1,n'}\}$. Now $X_i'$ represents the cost of $x_{i,l}$ should she be served by the facility placed in $C_i'$ of the changed instance (which now doesn't include her) and $X_{i-1}'$ her cost should she be served by the facility placed at $C_{i-1}'$ (which now includes her false declared location).
The expected cost of $x_{i,l}$ now becomes $\Exp(\min\{X_i',X_{i-1}'\})$. 

But, since $d_{i,1}\leq d_{i,2}\leq \ldots\leq d_{i,n-1}\leq d_{i-1,1}\leq \ldots \leq d_{i-1,n'}$, unless $d(x',x_{i,l})<d_{i,n-1} = D(C_i)$, we have that:
\begin{eqnarray*}
\Exp(\min\{X_i',X_{i-1}'\}) = \Exp(X_i') = \frac{d_{i,1}+ \ldots + d_{i,n-1}}{n-1} > \Exp(X_i)
\end{eqnarray*}
That means that $x_{i,l}$ cannot gain by this deviation unless $x'$ both belongs in $C_{i-1}'$ and $d(x',x_{i,l})< D(C_i)$. All we need to show now is that any such situation would result in a violation of the inter-cluster distance between $C_{i-1}'$ and $C_i'$ or between $C_{i-1}'$ and $C_{i-2}'$, guaranteed by the cluster-separation property and hence it would be caught by the mechanism's cluster-separation property verification step.

Specifically consider the distance of $x_{i,l}$ to her center $c_i$ of $C_i$ in the optimal clustering. We know that it must be $d(C_{i-1}, C_{i})\geq D(C_i)\cdot 1.6$, by Lemma~\ref{l:separated-clusters}, for the given stability factor of 5. But in order for this distance to be tight, it must be that $d(x_{i,l}, c_i) = 0.4\cdot \Diam(C_i)$ (see factor $c$ of proof of Lemma~\ref{l:separated-clusters} -due to stability properties, if $d(x_{i,l}, c_i) < 0.4\cdot \Diam(C_i)$ or $> 0.4\cdot \Diam(C_i)$, $d(C_{i-1}, C_{i})$ grows larger than  $D(C_i)\cdot 1.6$). Furthermore, in order for this distance to be tight, it must also be $d(c_{i-1},x_{i-1,r})<0.4\cdot \Diam(C_i)$ (since by stability it must be $d(C_{i-1}, C_{i}) = d(x_{i-1,r}, x_{i,l})>(\gamma -1)d(x_{i-1,r}, c_{i-1})$). 

Now, since it must be $d(x',x_{i,l})<\Diam(C_i)$ it will be $d(x',c_i)<1.4\Diam(C_i)$ and $d(x',x_{i-1,r})>0.6\Diam(C_i)$ (since $d(C_i,C_{i-1})>1.6\Diam(C_i)$ by Lemma~\ref{l:separated-clusters}). Finally we distinguish between two cases: 
\begin{description}
\item[Case 1:] $c_{i-1}\in C_{i-1}'$. Now notice that $d(x_{i,l},c_{i-1})>5d(x_{i,l},c_i)$ so $d(x_{i,l},c_{i-1})>2\Diam(C_i)$. Then $\Diam(C_{i-1}')\geq d(c_{i-1},x')> \Diam(C_i)$ (since $d(x',x_{i,l})<\Diam(C_i)$). But then $d(C_{i-1}', C_i')\leq d(x',c_i)\leq 1.4\cdot\Diam(C_i)$ which means that the cluster separation verification property of step 2 would be violated.
\item[Case 2:] $c_{i-1}\notin C_{i-1}'$. Then, in this edge case we notice it would be $d(C_{i-1}',C_{i-2}')\leq d(c_{i-1},x_{i-1,r})\leq 0.4\Diam(C_i)$. But $\Diam(C_{i-1}')\geq d(x_{i-1,r},x')\geq 0.6\Diam(C_i)$. Hence the verification property of step 2 is again violated between $C_{i-1}'$ and $C_{i-2}'$.
\end{description}

All we have to do to finish, is note that as $c_i$ moves to the right or to the left, $d(C_{i-1}, C_i)$ grows by a multiplicative factor $\gamma -1$ (=4) of $d(x_{i,l},c_i)$ (see proof of Leamma~\ref{l:separated-clusters}) and $d(x_{i,l},c_{i-1})$ by a multiplicative factor of $5$ (remember, it must be both $d(x_{i,l},c_{i-1})>5d(x_{i,l},c_i)$ and $d(x_{i,r},c_{i-1})>5d(x_{i,r},c_i)$). Which means that the above inequalities will still hold. \footnote{Notice here that while this property was a must-have for \rightm\ to work i.e. the mechanism wouldn't work if $x'$ both belongs in $C_{i-1}'$ and $d(x',x_{i,l})< D(C_i)$, here this might not the case. We can easily see this guarantees strategyproofness, but it might not be necessary which means the mechanism may work for smaller stability factors.}

For Case~2, why agent $x_{i,l}$ cannot gain by becoming a self serving cluster, we simply notice the following: her cost, should she not deviate, is at most $D(C_i)$ (see expected value from previous case). But, from Lemma~\ref{l:singleton_deviation_5} we know that $x_{i,l}$ must deviate by at-least $\geq D(C_i)$, for a stability factor of 5. So she cannot gain from this deviation\footnote{Again, while this property guarantees strategyproofness, it might not be necessary for example, we see that in one of the bad edge cases, where all agents of $C_i$ are gathered on $x_{i,r}$, with $c_i=x_{i,r}$ a stability of 3 would suffice to guarantee that $x_{i,l}$ needs to deviate by at-least $D(C_i)$ to become a self-serving cluster.}.

For Case~3, it is not hard to see that by merging all the agents in $C_i$ with agents $\notin C_i$, her expected cost can only increase. Furthermore, splitting the agents in $C_i$ would cause the cluster-separation property  verification step to identify the split (see the proof of the strategyproofness of the \rightm\ mechanism, in  Section~\ref{s:rightmost}) and remove all agents of $C_i$ from the game.
\end{proof}

\end{document}